\documentclass[a4paper,UKenglish,cleveref, autoref, thm-restate]{lipics-v2021}

\usepackage{microtype} 

\hideLIPIcs 
\title{The Complexity of Geodesic Spanners using Steiner Points}
\author{Sarita de Berg}{Department of Information and Computing Sciences, Utrecht University, The Netherlands}{S.deBerg@uu.nl}{https://orcid.org/0000-0001-5555-966X}{}

\author{Tim Ophelders}{Department of Information and Computing Sciences, Utrecht University, the Netherlands \and Department of Mathematics and Computer Science, TU Eindhoven, the Netherlands}{t.a.e.ophelders@uu.nl}{}{partially supported by the Dutch Research Council (NWO) under project no. VI.Veni.212.260.}
\author{Irene Parada}{Department of Mathematics, Universitat Politècnica de Catalunya, Spain}{irene.parada@upc.edu}{https://orcid.org/0000-0003-2401-8670}{is a Serra Húnter fellow and was supported by a Margarita Salas fellowship funded by the Ministry of Universities of Spain and the European Union (NextGenerationEU).}
\author{Frank Staals}{Department of Information and Computing Sciences, Utrecht University, The Netherlands}{F.Staals@uu.nl}{}{}
\author{Jules~Wulms}{Department of Mathematics and Computer Science, TU Eindhoven, the Netherlands}{J.H.H.M.Wulms@tue.nl}{https://orcid.org/0000-0002-9314-8260}{}

\authorrunning{S. de Berg, T. Ophelders, I. Parada, F. Staals, J. Wulms}

\Copyright{Sarita de Berg, Tim Ophelders, Irene Parada, Frank Staals, and Jules Wulms}

\ccsdesc[100]{Theory of computation~Computational Geometry} 

\keywords{spanner, simple polygon, polygonal domain, geodesic distance, complexity} 

\category{} 

\relatedversion{} 

 \acknowledgements{Work initiated at the 2023 AGA Workshop in Otterlo. We would like to thank Wouter Meulemans for his contributions in the initial phase of the project.}

\nolinenumbers 

\usepackage{graphicx}
\usepackage{amsfonts}
\usepackage[shortlabels]{enumitem}
\usepackage{xspace}

\usepackage{caption}
\usepackage[title]{appendix}
\usepackage{pifont}
\usepackage{color}

\graphicspath{ {./figures/} }

\newcommand {\mathset} [1] {\ensuremath {\mathbb {#1}}\xspace}
\newcommand {\R} {\mathset {R}}

\newcommand{\mkmcal}[1]{\ensuremath{\mathcal{#1}}\xspace}

\newcommand{\F}{\mkmcal{F}}
\newcommand{\G}{\mkmcal{G}}

\newcommand{\T}{\ensuremath{T}}
\newcommand{\SPT}{\ensuremath{\mathit{SPT}}}

\newcommand{\VC}{\ensuremath{W}\xspace}
\newcommand{\for}{\mkmcal{T}}
\newcommand{\Min}{\ensuremath{M_{\mathit{in}}}}
\newcommand{\Mout}{\ensuremath{M_{\mathit{out}}}}
\newcommand{\Pwithr}{\ensuremath{P'_r}}
\newcommand{\Pnotr}{\ensuremath{P'_{\neg r}}}
\newcommand{\steiner}{\mkmcal{S}\xspace}

\newcommand{\vertexcover}{\textsc{Vertex Cover}\xspace}
\newcommand{\steinerspanner}{\textsc{Steiner Spanner}\xspace}

\newcommand{\eps}{\ensuremath{\varepsilon}\xspace}
\newcommand{\etal}{et al.\xspace}

\begin{document}

\maketitle

\begin{abstract}
    A geometric $t$-spanner $\G$ on a set $S$ of $n$ point sites in a metric space $P$ is a subgraph of the complete graph on $S$ such that for every pair of sites $p,q$ the distance in $\G$ is a most $t$ times the distance $d(p,q)$ in $P$. We call a connection between two sites a \emph{link}. In some settings, such as when $P$ is a simple polygon with $m$ vertices and a link is a shortest path in $P$, links can consist of $\Theta (m)$ segments and thus have non-constant complexity. The spanner complexity is a measure of how compact a spanner is, which is equal to the sum of the complexities of all links in the spanner. In this paper, we study what happens if we are allowed to introduce $k$ Steiner points to reduce the spanner complexity. We study such Steiner spanners in simple polygons, polygonal domains, and edge-weighted trees.
    
    Surprisingly, we show that Steiner points have only limited utility. For a spanner that uses $k$ Steiner points, we provide an $\Omega(nm/k)$ lower bound on the worst-case complexity of any $(3-\eps)$-spanner, and an $\Omega(mn^{1/(t+1)}/k^{1/(t+1)})$ lower bound on the worst-case complexity of any $(t-\eps)$-spanner, for any constant $\eps \in (0,1)$ and integer constant $t \geq 2$. These lower bounds hold in all settings. Additionally, we show NP-hardness for the problem of deciding whether a set of sites in a polygonal domain admits a $3$-spanner with a given maximum complexity using $k$ Steiner points.
    
    On the positive side, for trees we show how to build a $2t$-spanner that uses $k$ Steiner points of complexity $O(mn^{1/t}/k^{1/t} + n \log (n/k))$, for any integer $t \geq 1$. We generalize this result to forests, and apply it to obtain a $2\sqrt{2}t$-spanner in a simple polygon with total complexity $O(mn^{1/t}(\log k)^{1+1/t}/k^{1/t} + n\log^2 n)$. When a link in the spanner can be any path between two sites, we show how to improve the spanning ratio in a simple polygon to $(2k+\varepsilon)$, for any constant $\varepsilon \in (0,2k)$, and how to build a $6t$-spanner in a polygonal domain with the same complexity.
\end{abstract}

\section{Introduction}

Consider a set $S$ of $n$ point \emph{sites} in a metric space $P$. In
applications such as (wireless) network
design~\cite{alzoubi03geomet_spann_wirel_ad_hoc_networ}, regression
analysis~\cite{gottlieb17effic_regres_metric_spaces_approx_lipsc_exten},
vehicle
routing~\cite{cohen-addad20light_spann_low_embed_effic,remy10euclid},
and constructing TSP tours~\cite{borradaile15near_stein}, it is
desirable to have a compact network that accurately captures the
distances between the sites in $S$. Spanners provide such a
representation. Formally, a \emph{geometric $t$-spanner} \G is a
subgraph of the complete graph on $S$, so that for every pair of sites
$p,q$ the distance $d_\G(p,q)$ in \G is at most $t$ times the distance
$d(p,q)$ in $P$~\cite{proximity_algorithms_book}. The quality of a
spanner can be expressed in terms of the \emph{spanning ratio} $t$ and
a term to measure how ``compact'' it is. Typical examples are the
\emph{size} of the spanner, that is, the number of edges of \G, its
weight (the sum of the edge lengths), or its diameter. Such spanners
are well studied~\cite{Arya95short_thin_lanky,
  survey_geometric_spanners,Chan15Doubling_spanners,Elkin15optimal_short_thin_lanky}. For
example, for point sites in $\R^d$ and any constant $\eps > 0$ one can
construct a $(1+\eps)$-spanner of size
$O(n/\eps^{d-1})$~\cite{book_spanners}. Similar results exist for more
general
spaces~\cite{RoutingInDoublingMetrics,bounded_doubling_optimal_dynamic,
  FastConstructionDoublingMetrics}. Furthermore, there are various
spanners with other desirable spanner properties such as low maximum
degree, or
fault-tolerance~\cite{Bose05Bounded_degree_low_weight,Le19truly_opt_Euclidean,Levcopoulos98_fault_tolerant,book_spanners}.

When the sites represent physical locations, there are often other
objects (e.g. buildings, lakes, roads, mountains) that influence the
shortest path between the sites. In such settings, we need to
explicitly incorporate the environment. We consider the case where
this environment is modeled by a polygon $P$ with $m$ vertices, and
possibly containing holes. The distance between two points
$p,q \in P$ is then given by their \emph{geodesic distance}: the
length of a shortest path between $p$ and $q$ that is fully contained
in $P$. This setting has been considered before. For example, Abam,
Adeli, Homapou, and Asadollahpoor~\cite{SpannerPolygonalDomain}
present a $(\sqrt{10} +\eps)$-spanner of size $O(n \log^2 n)$
when $P$ is a simple polygon, and a $(5 + \eps)$-spanner of
size $O(n \sqrt{h} \log^2n)$ when the polygon has $h>1$ holes. Abam,
de Berg, and Seraji~\cite{SpannerPolyhedralTerrain} even obtain a
$(2+\eps)$-spanner of size $O(n\log n)$ when $P$ is actually
a terrain. To avoid confusion between the edges of $P$ and the edges of \G, we will from hereon use the term \emph{links} to refer to the
edges of~\G.

As argued by de Berg, van Kreveld, and Staals~\cite{complexity_spanners}, 
each link in a geodesic spanner may
correspond to a shortest path containing $\Omega(m)$ polygon
vertices. Therefore, the \emph{spanner complexity}, defined as the
total number of line segments that make up all links in the spanner,
more appropriate measures how compact a geodesic spanner is. In this definition, a line segment that appears in multiple links is counted multiple times: once for each link it appears in.
The above
spanners of~\cite{SpannerPolygonalDomain,SpannerPolyhedralTerrain} all have worst-case complexity $\Omega(mn)$, hence de Berg, van Kreveld, and Staals
present an algorithm to construct a $2\sqrt{2}t$-spanner in a simple polygon with
complexity $O(mn^{1/t} + n\log^2 n)$, for any integer $t\geq 1$. By relaxing the restriction of links being shortest paths to any path between two sites, they obtain, for any constant $\eps \in (0,2t)$, a \emph{relaxed} geodesic $(2t+\eps)$-spanner in a simple polygon, or a relaxed geodesic $6t$-spanner in a polygon with holes, of the same complexity. These complexity bounds are still relatively high. 
De Berg, van Kreveld, and Staals~\cite{complexity_spanners} also
show that these results are almost tight. 
In particular, for sites in
a simple polygon, any geodesic $(3-\eps)$-spanner has worst-case complexity
$\Omega(nm)$, and for any constant $\varepsilon \in (0,1)$ and integer
constant $t \geq 2$, a $(t-\varepsilon)$-spanner has worst-case complexity
$\Omega(mn^{1/(t-1)} + n)$.

\subparagraph{Problem Statement.} A very natural question is then if
we can reduce the total complexity of a geodesic spanner by allowing
\emph{Steiner points}. That is, by adding an additional set \steiner of
$k$ vertices in \G, each one corresponding to a (Steiner) point in
$P$. For the original sites $p,q \in S$ we still require that their
distance in \G is at most $t$ times their distance in $P$, but the graph distance from a Steiner
point $p' \in \steiner$ to any other site is unrestrained. Allowing
for such Steiner points has proven to be useful in reducing the
weight~\cite{bhore21light_euclid_stein_spann_plane,Elkin_low_light_Steiner} and size~\cite{Le19truly_opt_Euclidean} of spanners. In our setting, it
allows us to create additional ``junction'' vertices, thereby allowing
us to share high-complexity subpaths. See
Figure~\ref{fig:example_Steiner_spanner} for an illustration. Indeed,
if we are allowed to turn every polygon vertex into a Steiner point,
Clarkson~\cite{Clarkson87_approx_sp} shows that, for any $\eps > 0$,
we can obtain a $(1+\eps)$-spanner of complexity
$O((n+m)/\eps)$. However, the number of polygon vertices $m$ may be
much larger than the number of Steiner points we can afford. Hence, we
focus on the scenario in which the number of Steiner points $k$ is
(much) smaller than $m$ and~$n$.

\begin{figure}
    \centering
    \includegraphics{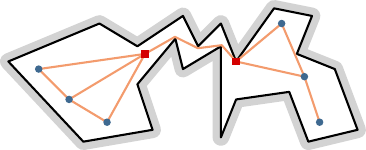}
    \caption{A spanner in a simple polygon that uses two Steiner points (red squares). By adding the two Steiner points, the spanner has a small spanning ratio and low complexity, as we no longer need multiple links that pass through the middle section of~$P$.}
    \label{fig:example_Steiner_spanner}
\end{figure}

\subparagraph{Our Contributions.} Surprisingly, we show that in this
setting, Steiner
points have only limited utility. In some cases, even a single Steiner point allows us to improve the complexity by a linear factor. However, we show that such improvements are not possible in general. First of all, we show that computing
a minimum cardinality set of Steiner points for sites in a polygonal
domain that allow for a $3$-spanner of a certain complexity is NP-hard. Moreover, we show that there is a set of $n$ sites
in a simple polygon with $m=\Omega(n)$ vertices for which any
$(2-\eps)$-spanner (with $k < n/2$ Steiner points) has complexity $\Omega(mn^2/k^2)$. Similarly, we give a $\Omega(mn/k)$ and 
$\Omega(mn^{1/(1+t)}/k^{1/(1+t)})$ lower bound on the complexity of a
$(3-\eps)$- and $(t-\eps)$-spanner with $k$ Steiner points. Hence, these results dash
our hopes for a near linear complexity spanner with ``few'' Steiner
points and constant spanning ratio.

These lower bounds actually hold in a more restricted setting. Namely,
when the metric space is simply an edge-weighted tree that has $m$
vertices, and the $n$ sites are all placed in leaves of the tree.  In
this setting, we show that we can efficiently construct a spanner
whose complexity is relatively close to optimal. In particular, our
algorithm constructs a $2t$-spanner of complexity
$O(mn^{1/t}/k^{1/t}+n\log(n/k))$. The main idea is to partition the
tree into $k$ subtrees of roughly equal size, construct a $2t$-spanner
without Steiner points on each subtree, and connect the spanners of
adjacent trees using Steiner points. The key challenge that we tackle,
and one of the main novelties of the paper, is to make sure that each
subtree contains only a constant number of Steiner points. We
carefully argue that such a partition exists, and that we can
efficiently construct it. Constructing the spanner takes
$O(n\log(n/k)+m+K)$ time, where $K$ is the output complexity. This
output complexity is either the size of the spanner ($O(n\log(n/k))$),
in case we only wish to report the endpoints of the links, or the
complexity, in case we wish to explicitly report the shortest paths
making up the links. An extension of this algorithm allows us to deal
with a forest as well.

This algorithm for constructing a spanner on an edge-weighted tree
turns out to be the crucial ingredient for constructing low-complexity
spanners for point sites in polygons. In particular, given a set of
sites in a simple polygon $P$, we use some of the techniques developed
by de Berg, van Kreveld, and Staals~\cite{complexity_spanners} to
build a set of trees whose leaves are the sites, and in which the
distances in the trees are similar to the distances in the polygon. We
then construct a $2t$-spanner with $k$ Steiner points on this forest
of trees using the above algorithm, and argue that this actually
results into a $2\sqrt{2}t$-spanner with respect to the distances in
the polygon. The main challenge here is to argue that the links used
still have low complexity, even when they are now embedded in the
polygon. We prove that the spanner (with respect to the polygon) has
complexity $O(mn^{1/t}(\log k)^{1+1/t}/k^{1/t} + n\log^2 n)$, and can
be constructed in time $O(n\log^2 n + m\log n + K)$. If we allow a link in the spanner to be any path between two sites (or Steiner points), then we obtain for any constant $\eps \in (0,2k)$ a relaxed $(2t+\eps)$-spanner of the same complexity. For $k = O(1)$
our spanners thus match the results of de Berg, van Kreveld, and
Staals~\cite{complexity_spanners}. Finally, we extend these results to
polygonal domains, where we construct a similar complexity relaxed $6t$-spanner in $O(n\log^2n + m \log n \log m + K)$ time.

\subparagraph{Organization.} We start with our results on
edge-weighted trees in
Section~\ref{sec:Steiner_spanners_for_trees}. To get a feel for the
problem, we first establish lower bounds on the spanner complexity in Section~\ref{sub:tree_lower_bounds}. In Section~\ref{sub:tree_spanner} we present the
algorithm for efficiently constructing a low complexity $2t$-spanner.
We extend this result to a forest in Section~\ref{sub:forest_spanner}. 
In Section~\ref{sec:spanners_in_a_simple_polygon}, we show
how to use these results to obtain a $2\sqrt{2}t$-spanner for sites in
a simple polygon $P$. In Section~\ref{sec:polygonal_domain} we further extend our algorithms to the most
general case in which~$P$ may even have holes. 
In Section~\ref{sub:np-hardness} we show that computing a minimum
cardinality set of Steiner points with which we can simultaneously
achieve a particular spanning ratio and maximum complexity is
NP-hard. Finally, in Section~\ref{sec:future_work} we pose some remaining open questions. 

\section{Steiner spanners for trees}
\label{sec:Steiner_spanners_for_trees}

In this section, we consider spanners on an edge-weighted rooted tree $\T$. We allow only positive weights. The goal is to construct a $t$-spanner on the leaves of the tree that uses $k$ Steiner points, i.e. the set of sites $S$ is the set of leaves. 
We denote by $n$ the number of leaves and by $m$ the number of vertices in $\T$. The complexity of a link between two sites (or Steiner points) $p,q \in \T$ is the number of edges in the shortest path $\pi(p,q)$, and the distance $d(p,q)$ is equal to the sum of the weights on this (unique) path. We denote by~$\T(v)$ the subtree of $\T$ rooted at vertex~$v$. For an edge $e \in \T$ with upper endpoint~$v_1$ (endpoint closest to the root) and lower endpoint~$v_2$, we denote by $\T(e) := \T(v_2) \cup \{e\}$ the subtree of $T$ rooted at $v_1$.

The Steiner points are not restricted to the vertices of $T$, but can lie anywhere on the tree. To be precise, for any $\delta \in (0,1)$ a Steiner point $s$ can be placed on an edge $(u,v)$ of weight $w$. This edge is then replaced by two edges $(u,s)$ and $(s,v)$ of weight $\delta w$ and $(1-\delta)w$.
Observe that this increases the complexity of a spanner on $T$ by at most a constant factor as long as there are at most a constant number of Steiner points placed on a single edge. The next lemma states that it is indeed never useful to place more than one Steiner point on the interior of an edge.
\begin{lemma}\label{lem:steiner_interior_edge}
    If a $t$-spanner $\G$ of a tree $T$ has more than one Steiner point on the interior of an edge $e = (u,v)$, then we can modify $\G$ to obtain a $t$-spanner $\G'$ that has no Steiner points on the interior of $e$ without increasing the complexity and number of Steiner points.
\end{lemma}
\begin{proof}
    Let $\steiner$ denote the set of Steiner points of $\G$ and let $\steiner(e) \subseteq \steiner$ the subset of Steiner points that lie on $e$. We assume that each Steiner point is used by a path $\pi_\G(p,q)$ for some sites $p,q$, otherwise we can simply remove it. We define the set of Steiner points of $\G'$ as $\steiner'=(\steiner \setminus \steiner(e)) \cup \{u,v\}$. Observe that $|\steiner'| \leq |\steiner|$. To obtain $\G'$, we replace each link $(p,s)$ with $s \not\in \steiner'$ by $(p,u)$ if $(p,s)$ intersects $u$ and by $(p,v)$ if $(p,s)$ intersects $v$. Links between Steiner points on $e$ are simply removed. Finally, we add the link $(u,v)$ to $\G'$.
    
    We first argue that the spanning ratio of $\G'$ is as most the spanning ratio of $\G$. Consider a path between two sites $p,q$ in $\G$. If this path still exists in $\G'$, then $d_\G(p,q) = d_{\G'}(p,q)$. If not, then the path must visit $e$. Let $(p_1,s_1)$ and $(p_2, s_2)$ denote the first and last link in the path that connect to a Steiner point in the interior of $e$ (possibly $s_1 = s_2$). If $\pi(p,q)$ does not intersect the open edge $e$, then these links are replaced by $(p_1,u)$ and $(p_2,u)$ (or symmetrically by $(p_1,v)$ and $(p_2,v)$) in $\G'$.
    This gives a path in $\G'$ via $u$ such that $d_{\G'}(p,q) < d_\G(p,q)$. If $\pi(p,q)$ does intersect $e$, i.e. $p$ and $q$ lie on different sides of $e$, then, without loss of generality, the links $(p_1,s_1)$ and $(p_2, s_2)$ are replaced by $(p_1,u)$, $(u,v)$, and $(p_2, v)$. Again, this gives a path in $\G'$ such that $d_{\G'}(p,q) \leq d_\G(p,q)$.
    
    Finally, what remains is to argue that the complexity of the spanner does not increase. Each link that we replace intersects either $u$ or $v$, thus replacing this link by a link up to $u$ or $v$ reduces the complexity by one. Because each Steiner point on $e$ occurs on at least one path between sites in $\G$, we replace at least two links. This decreases the total complexity by at least two, while including the edge $(u,v)$ increases the complexity by only one.
\end{proof}

\begin{corollary}\label{cor:interior_steiner_points}
    Any spanner $\G$ on a tree $T$ can be modified without increasing the spanning ratio and complexity such that no edge contains more than one Steiner point in its interior.
\end{corollary}

\subsection{Complexity lower bounds}
\label{sub:tree_lower_bounds}
In this section, we provide several lower bounds on the worst-case complexity of any $(t-\eps)$-spanner that uses $k$ Steiner points, where $t$ is an integer constant and $\eps \in (0,1)$. 
When Steiner points are not allowed, any $(2-\eps)$-spanner in a simple polygon requires $\Omega(n^2)$ edges~\cite{SpannerPolygonalDomain} and $\Omega(mn^2)$ complexity. If we allow a larger spanning ratio, say $(3-\eps)$ or even $(t-\eps)$, the worst-case complexity becomes $\Omega(mn)$ or $\Omega(mn^{1/(t-1)})$, respectively~\cite{complexity_spanners}. As the polygons used for these lower bounds are very tree-like, these bounds also hold in our tree setting. Next, we show how much each of these lower bounds is affected by the use of $k$ Steiner points.

    \begin{figure}
        \centering
        \includegraphics[page=11]{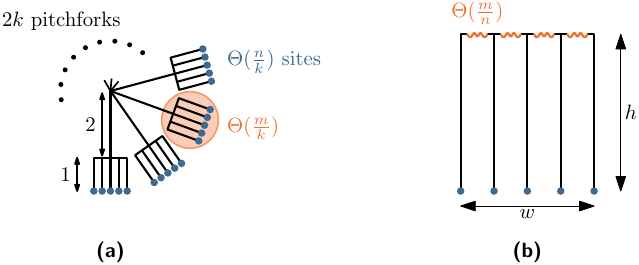}
        \caption{\textbf{\textsf{(a)}} Our construction for an $\Omega(mn^2/k^2)$ lower bound on the complexity of any $(2-\eps)$-spanner. \textbf{\textsf{(b)}} A more detailed version of the comb of a pitchfork highlighted in the orange disk, which is also used for our $\Omega(mn^{1/(t+1)}/k^{1/(t+1)})$ lower bound on the complexity of any $(t-\eps)$-spanner. 
        }
        \label{fig:lower-bound-2-spanner}
    \end{figure}

\begin{restatable}{lemma}{lowerboundtwospanner}
  \label{lem:lowerbound_2spanner}
  For any constant $\eps \in (0,1)$, there exists an edge-weighted tree $\T$ for which any $(2-\varepsilon)$-spanner using $k < n/2$ Steiner points has complexity~$\Omega(mn^2/k^2)$.
\end{restatable}
\begin{proof}
    The tree~$T$ in our construction is a star of `pitchforks', with long handles and short teeth, that attach at the handles. See Figure~\ref{fig:lower-bound-2-spanner}(a) and (b) for an illustration of the construction. Precisely, $\T$ consists of $2k$ pitchforks, that have handles of weight~$2$, a comb of weight-$1$ teeth, and all other edges have much smaller weight, defined later. All pitchforks have an (almost) equal number of teeth, at the end of which are the sites, and the top of each comb has high complexity, as in Figure~\ref{fig:lower-bound-2-spanner}(b), leading to $\Theta(n/k)$ sites and $\Theta(m/k)$ vertices per pitchfork. 

    For a site~$p$ in one pitchfork, consider the distance to a site~$q$ in the same pitchfork. First, observe that for every~$\eps$, we can set the remaining weights to be small enough such that $d(p,q)$ approaches $2$. 
    It follows that the path from~$p$ to~$q$ in any $(2-\eps)$-spanner~$\G$ does not visit any other site, as this would result in $d_\G(p,q)$ violating the spanning ratio. 
    
   To analyze the complexity of a $(2-\eps)$-spanner~$\G$, consider the at least $k$ pitchforks that do not contain a Steiner point (a possible Steiner point at the center of the star is not part of any pitchfork).
   Inside each of these pitchforks, the complete graph on the sites must be a subgraph of $\G$. We thus have $\Omega(n^2/k^2)$ links of complexity $\Omega(m/k)$. 
    As there are $\geq k$ of these pitchforks, this leads to a Steiner spanner with~$\Omega(mn^2/k^2)$ complexity.
\end{proof}

\begin{restatable}{lemma}{lowerboundthreespanner}
  \label{lem:lowerbound_3spanner}
  For any constant $\eps \in (0,1)$, there exists an edge-weighted tree $\T$ for which any $(3-\varepsilon)$-spanner using $k < n/2$ Steiner points has complexity $\Omega(mn/k)$.
\end{restatable}
\begin{proof}
    \begin{figure}
        \centering
        \includegraphics[page=3]{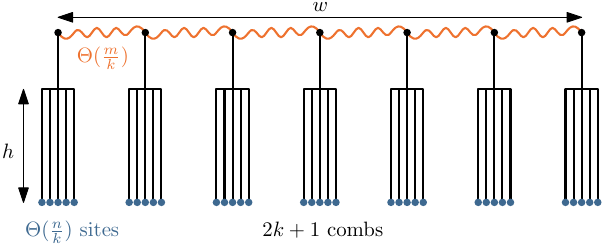}
        \caption{Our construction for a $\Omega(nm/k)$ lower bound on the complexity of any $(3-\eps)$-spanner.
        }
        \label{fig:lower-bound-3-spanner-div-k}
    \end{figure}
    We generalize the construction for the lower bound on the complexity of any $(3-\eps)$-spanner~\cite{complexity_spanners}, but instead of a polygon we work with a more restrictive setting of an edge-weighted tree. In the $(3-\eps)$ lower bound, the polygon consists of `combs' with long thin teeth, and we use a congruent structure in our trees. See Figure~\ref{fig:lower-bound-3-spanner-div-k} for an illustration of the construction.
    Our tree~$\T$ consist of a sequence of $2k+1$ of these combs,
    which are connected by their short handles in series using high complexity paths, which we will refer to as corridors. 
    Let $w$ be the total weight of the corridors.
    Remember that the spanner will be constructed on the leaves of~$\T$, which are at the bottom of each tooth of each comb. For each tooth we define the edge weight as $h \gg w$. All other edges have negligible weight.
    When $w$ approaches $0$, the distance between any two sites $p,q$ approaches $2h$. 
    So, in any $(3-\eps)$-spanner $\G$ there can be at most one other site on the path from $p$ to $q$.
    
    Let $M = \Theta(m/k)$ denote the complexity of a single corridor. 
    For any particular Steiner point, the set of sites that can be reached by a path of complexity $<M/2$ cannot lie in different combs.
    Hence, after placing $k$ Steiner points, at least $k+1$ of the $2k+1$ combs have no site that lies within complexity $<M/2$ of any Steiner point.

    Consider the $k+1$ combs that have no path of complexity $<M/2$ to any Steiner point. 
    Let $S_1,\dots,S_{k+1}$ be the sets of sites that lie in these $k+1$ combs, respectively. Observe that $|S_i| = \Theta(n/k)$.
    We say that a link with one endpoint at a site in $S_i$ \emph{leaves its comb} if its other endpoint is either a site in $S\setminus S_i$, or a Steiner point.
    Any link that leaves its comb has complexity $\Omega(M)$.

    We will show that there are $\Omega(n)$ links that leave their comb.
    There are two cases: 
    (1)~every site in $\bigcup_i S_i$ has a link that leaves its comb, 
    or (2)~for some $i$ and some site $p\in S_i$, the site $p$ has no link that leaves its comb.
    In case~(1), there are at least $\frac{1}{2}\sum_{i=1}^{k+1} |S_i|=\Omega((k+1)\frac{n}{2k+1})=\Omega(n)$ links that leave a comb.
    In case~(2), consider a site $q$ in $\bigcup_{j\neq i} S_j$. The path from $q$ to $p$ can visit at most one other site. Since $p$ has no link that leaves its comb, $q$ must have a link to a site in $S_i$, which leaves its comb.
    There are at least $\sum_{j\neq i}|S_j|=\Omega(k\frac{n}{2k+1})=\Omega(n)$ such links.
    In both cases, we have $\Omega(n)$ links of complexity $\Omega(M) = \Omega(m/k)$, so the total complexity is~$\Omega(mn/k)$.
\end{proof}

\begin{lemma}
\label{lem:lowerbound_tspanner}
For any constant $\eps \in (0,1)$ and integer constant $t \geq 2$, there exists an edge-weighted tree $\T$ for which any $(t-\varepsilon)$-spanner using $k < n$ Steiner points has complexity $\Omega(mn^{1/(t+1)}/k^{1/(t+1)})$.
\end{lemma}

    Before we prove Lemma~\ref{lem:lowerbound_tspanner}, we first discuss a related  result in a simpler metric space.
    Let~$\vartheta_n$ be the 1-dimensional Euclidean metric space with $n$ points $v_1,\dots,v_n$ on the $x$-axis at $1,2,\dots,n$. 
    A link $(v_i,v_j)$ has complexity $|i-j|$. Dinitz, Elkin, and Solomon~\cite{shallow_low_light_trees} give a lower bound on the total complexity of any spanning subgraph of $\vartheta_n$, given that the link-radius is at most $\rho$. The link-radius (called hop-radius in~\cite{shallow_low_light_trees}) $\rho(G,r)$ of a graph $G$ with respect to a root $r$ is defined as the maximum number of links needed to reach any vertex in $G$ from $r$. The link-radius of $G$ is then $\min_{r\in V} \rho(G,r)$. The link-radius is bounded by the link-diameter, which is the minimum number of links that allow reachability between any two vertices.
    \begin{lemma}[Dinitz \etal~\cite{shallow_low_light_trees}] \label{lem:lower_bound_vartheta}
    For any sufficiently large integer $n$ and positive integer $\rho < \log n$, any spanning subgraph of $\vartheta_n$ with link-radius at most $\rho$ has complexity 
    $\Omega(\rho \cdot n^{1+1/\rho})$.
    \end{lemma}

\begin{proof}[Proof of Lemma~\ref{lem:lowerbound_tspanner}]
    Consider the tree construction illustrated in Figure~\ref{fig:lower-bound-2-spanner}(b). 
    This edge-weighted tree~$T$ has the shape of a comb of width $w$ and height $h$ with $n$ teeth separated by corridors of complexity $M=\Theta(m/n)$ each.
    Each leaf at the bottom of a comb tooth is a site. 

    Any spanning subgraph of $\vartheta_n$ of complexity $C$ and link-diameter $\rho$ is in one-to-one correspondence with a $(\rho + 1 - \eps)$-spanner of complexity $C \cdot m/n$ in $T$~\cite{complexity_spanners}. 
    Lemma~\ref{lem:lower_bound_vartheta} then implies that any $(t-\eps)$-spanner has worst-case complexity $\Omega(mn^{1/(t-1)})$.

    \begin{figure}
        \centering
        \includegraphics[page=10]{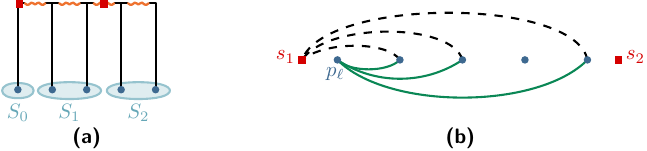} 
        \caption{\textbf{\textsf{(a)}} The sets $S_i$ defined by the two (red square) Steiner points. \textbf{\textsf{(b)}} For any spanner on~$S_1$, every link to a Steiner
          point can be replaced by a link of smaller complexity, while increasing the spanning ratio by at most one. Here, the dashed links can be replaced by the green links.}
        \label{fig:replacing_edges}
    \end{figure}

    When a set $\steiner$ of $k$ Steiner points is introduced, we consider the at most $k+1$ sets $S_0,\dots,S_{k}$ of consecutive sites that have no Steiner point between them; see Figure~\ref{fig:replacing_edges}(a). 
    We can replace any link $(p,q)$ where $p,q \in S \cup \steiner$ and $\pi(p,q)$ intersects a Steiner point $s$ by the links $(p,s)$ and $(s,q)$. Corollary~\ref{cor:interior_steiner_points} implies that this increases the complexity by only a constant factor. From now on, we thus assume there are no such links.
    We claim the following for any $(t-\eps)$-spanner $\G$ on $S = S_0 \cup \dots \cup S_k$.

    \begin{claim}
    Let $C_i$ be the complexity of the subgraph of $\G$ induced by $S_i$
    and the at most two Steiner points $s_{\ell}$ and $s_r$ bounding $S_i$ from the left and right, respectively.
    Then, we can construct a $(t+2-\eps)$-spanner $\G'_i$ on $S_i$ that has complexity at most $C_i$.  
    \end{claim}
    \begin{proof}[Proof of Claim]
    Let $p_\ell$ and $p_r$ denote the leftmost and rightmost site in $S_i$.
    We replace each link $(p,s_{\ell})$ (or $(p,s_r)$), $p \in S_i$, by the link $(p,p_\ell)$ (resp. $(p,p_r)$). 
    If there is a link
    $(s_{\ell},s_r)$, it is replaced by $(p_\ell,p_r)$. 
    Any path in $\G$ between $p,q \in S_i$ that visits
    either $s_{\ell}$ and/or $s_r$ corresponds to a path via $p_r$ and/or
    $p_\ell$ in $\G_i'$. 
    The length of the path increases by
    $2h$
    when visiting $p_r$ or $p_\ell$, so by at most 4$h$ when
    visiting both. 
    As $d(p,q) \geq 2h$, the spanning ratio increases by at most two. 
    \end{proof}

    These changes in the spanner only decrease the complexity of
    the subspanner on $S_i$. 
    Notice also that if we apply them to each of the sets $S_i$, each link of \G is changed by only one of the subspanners $\G'_i$.
    Thus, we consider the minimum complexity of any
    $(t+2-\eps)$-spanner on these sites. By applying
    Lemma~\ref{lem:lower_bound_vartheta}, we find that the worst-case
    complexity of any $(t+2-\eps)$-spanner on these $|S_i|$ sites is
    $\Omega(m/n \cdot |S_i|^{1 + 1/(t+1)})$. 
    The complexity of \G is at least the sum of the complexities of these $\G'_i$ spanners over all $S_i$, so 
    $\frac{m}{n}\sum_{i=0}^{k} \Omega\left(|S_i|^{1 + 1/(t+1)}\right)$, where $\sum_{i=0}^k |S_i| = \Theta(n)$.
    Using a logarithmic transformation and induction, we see that this sum is minimized when $|S_i| = \Theta(n/k)$ for all $i \in
    {0,\dots,k}$.
    So,
    \begin{align*}\label{eq:lower_bound_t+1}
        \frac{m}{n}\sum_{i=0}^{k} \Omega\left(|S_i|^{1 + 1/(t+1)}\right) &\geq \frac{m}{n}\sum_{i=0}^{k} \Omega\left((n/k)^{1 + 1/(t+1)}\right) 
        = \Omega\left(mn^{1/(t+1)}/k^{1/(t+1)}\right). \qedhere
    \end{align*}
\end{proof}

\subsection{A low complexity Steiner spanner}
\label{sub:tree_spanner}
In this section, we describe how to construct low complexity spanners for edge-weighted trees. The goal is to construct a $2t$-spanner of complexity $O(mn^{1/t}/k^{1/t} + n\log n)$ that uses at most $k$ Steiner points. We first show that the spanner construction for a simple polygon of~\cite{complexity_spanners} can be used to obtain a low complexity spanner for a tree (without Steiner points).

\begin{restatable}[de Berg \etal~\cite{complexity_spanners}]{lemma}{tSpannerTree}\label{lem:t-spanner-tree}
    For any integer $t\geq 1$, we can build a $2t$-spanner for a tree $\T$ of size $O(n\log n)$ and complexity $O(mn^{1/t} + n\log n)$ in $O(n\log n + m)$ time.
\end{restatable}
\begin{proof}
    De Berg, van Kreveld, and Staals consider constructing a spanner for some weighted 1-dimensional space, which is related to a spanner in a polygon. They essentially show how to construct a $2t$-spanner of complexity $O(mn^{1/t} + n \log n)$ on a shortest path tree. Their approach also applies to any edge-weighted rooted tree as we defined here, and we briefly sketch the construction.
    First, an ordering on the sites is defined based on an in-order traversal of the tree. Then, we find an edge $e$ such that roughly half of the sites lie in each subtree after removing $e$. Fix $N = n^{1/t}$. We partition the sites into $\Theta(N)$ groups based on the ordering. Each group is again partitioned into $\Theta(N)$ groups, etc. At each level, we select the site that is closest to $e$ as the group center $c$ for each group, and add a link to the spanner from $c$ to the center of its parent group. This ensures that the spanning ratio holds for any pair of sites $p,q$ with $p$ in the one subtree and $q$ in the other. To also obtain good paths between sites in the same subtree, we recurse on both subtrees.
\end{proof}

\subparagraph{Spanner construction.} Given an edge-weighted tree $T$, to construct a Steiner spanner $\G$ for $T$, we start by partitioning the sites in $k$ sets $S_1,\dots S_{k}$ by an in-order traversal of the tree. The first $\lceil n/k\rceil$ sites encountered are in $S_1$, the second $\lceil n/k\rceil$ in $S_2$, etc. 
After this, the sites are reassigned into $k$ new disjoint sets $S_1',\dots S_k'$. For each of these sets, we consider a subtree $\T_i'\subseteq \T$ whose leaves are the set~$S_i'$. There are four properties that we desire of these sets and their subtrees.
\begin{enumerate}
    \item The size of $S_i'$ is $O(n/k)$.
    \item The trees $\T_i'$ cover $\T$, i.e. $\bigcup_i \T_i' = \T$.
    \item The trees $\T_i'$ are disjoint apart from Steiner points.
    \item Each tree $\T_i'$ contains only $O(1)$ Steiner points.
\end{enumerate}
As we prove later, these properties ensure that we can construct a spanner on each subtree $\T_i'$ to obtain a spanner for $\T$. We obtain such sets $S_i'$ and the corresponding trees $\T_i'$ as follows.

We color the vertices and edges of the tree $\T$ using $k$ colors $\{1,\dots,k\}$ in two steps. In this coloring, an edge or vertex is allowed to have more than one color. First, for each set $S_i$, we color the smallest subtree that contains all sites in $S_i$ with color $i$. 
After this step, all uncolored vertices have only uncolored incident descendant edges.
Second, we color the remaining uncolored edges and vertices. These edges and their (possibly already colored) upper endpoints are colored in a bottom-up fashion. We assign each uncolored edge and its upper endpoint the color with the lowest index~$i$ that is assigned also to its lower endpoint. 
After coloring $\T$, we for $i \in \{1,\dots, k\}$ place a Steiner point $s_i$ at the root of tree $\T_i$ formed by all edges and vertices of color $i$.
This may place multiple Steiner points at the same vertex.
We may abuse notation, and denote by $s_i$ the vertex occupied by Steiner point~$s_i$.

For each Steiner point $s_i$, we define a subtree $\T_i' \subseteq \T$. The sites in $\T_i'$ will be the set~$S_i'$. The tree $\T_i'$ is a subtree of $\T(s_i)$. When $s_i$ is the only Steiner point at the vertex, then $\T_i'= \T(s_i) \setminus \bigcup_j (\T(s_j) \setminus \{s_j\})$ for $s_j$ a descendant of $s_i$. In other words, we look at the tree rooted at $s_i$ up to and including the next Steiner points, see Figure~\ref{fig:tree-spanner}(a). When $s_i$ is not the only Steiner point at the vertex, we include only subtrees $\T(e)$ of $s_i$ (up to the next Steiner points) that start with an edge $e$ that has color $i$ and no color $j > i$. See Figure~\ref{fig:tree-spanner}(b). Whenever $s_i$ has the lowest or highest index of the Steiner points at $s_i$, we also include all $\T(e')$ that start with an edge $e'$ of color $j < i$ or $j > i$, respectively. This generalizes the scheme for when $s_i$ is the only Steiner point at the vertex. 

By creating $\T_i'$ in this way, $s_i$ is not a leaf of $\T_i'$. 
We therefore adapt $\T_i'$ by adding an edge of weight zero between the vertex at $s_i$ and a new leaf corresponding to $s_i$. 
On each subtree $\T_i'$, we construct a $2t$-spanner using the algorithm of Lemma~\ref{lem:t-spanner-tree}. These $k$ spanners connect at the Steiner points, which we formally prove in the spanner analysis.

\begin{figure}
    \centering
    \includegraphics[page=3]{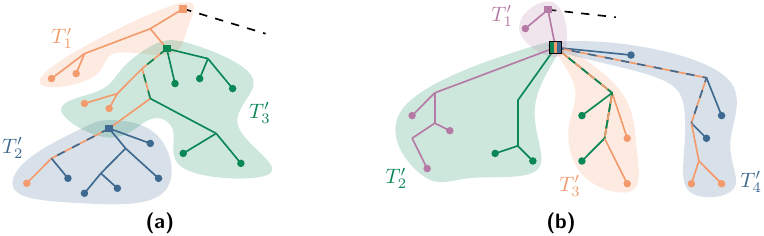}
    \caption{The tree $\T_i$ is the subtree whose edges and vertices have color $i$. A Steiner point (square) is placed at the root of $\T_i$. The shaded areas show the trees $\T_i'$. The examples show the case when the Steiner points are \textbf{\textsf{(a)}} at different vertices or \textbf{\textsf{(b)}} share a vertex.}
    \label{fig:tree-spanner}
\end{figure}

\subparagraph{Analysis.} To prove that $\G$ is indeed a low complexity $2t$-spanner for $T$, we first show that the four properties stated before hold for $S_i'$ and $\T_i'$. We often apply the following lemma, that limits the number of colors an edge can be assigned by our coloring scheme.
\begin{lemma}\label{lem:two_colors}
    An edge can have at most two colors.
\end{lemma}
\begin{proof}
    First, observe that an edge can receive more than one color only in the first step of the coloring. Suppose for contradiction that there is an edge $e$ in $\T$ that has three colors $i<j<\ell$. Let $v$ be the lower endpoint of $e$. Then there must be three sites $p_i \in S_i$, $p_j \in S_j$, $p_\ell \in S_\ell$ in $\T(v)$. Because these sets are defined by an in-order traversal, $p_i$ must appear before $p_j$ in the traversal. Similarly, $p_j$ appears before $p_\ell$. Additionally, there must be a site $p_j' \in S_j$ in $\T \setminus \T(v)$, otherwise the color $j$ would not be assigned to $e$. 
    The site $p_j'$ must appear before $p_i$ or after $p_\ell$ in the traversal. In the first case, $p_i$ must be in $S_j$ as it appears between two sites in $S_j$. In the second case, we find $p_\ell \in S_j$, also giving a contradiction.
\end{proof}

\begin{restatable}{lemma}{sizeSi}\label{lem:size_S_i}
    The size of $S_i'$ is $O(n/k)$.
\end{restatable}
\begin{proof}
     Each set $S_i'$ can contain sites from an original set $S_j$, $j \neq i$, only if there is a site $p \in S_j$ in $\T(s_i)$ and the Steiner point $s_j$ lies at or above $s_i$. Lemma~\ref{lem:two_colors} implies that there are at most two such sets that have a Steiner point above $s_i$. We claim that there can be at most one such set $S_j$ that has a site $p \in S_i'$ for which $s_i = s_j$. Note that $j < i$ by definition of $\T_i'$.   
     Suppose there is another set $S_x$ for which there is also a site $q \in S_i'$ and $s_x = s_i$. Consider the first edge $e$ and $e'$ on the paths $\pi(s_i, p)$ and $\pi(s_i, q)$, respectively. Both $e$ and $e'$ have color $i$, otherwise $p$ or $q$ would not be in $S_i'$. As $s_i = s_j = s_x$, $e$ must also have color $j$ and $e'$ must also have color $x$. Lemma~\ref{lem:two_colors} now implies that $e \neq e'$. However, there must be a site of $S_i$ in both $\T(e)$ and $\T(e')$, which is impossible as $j < i$ and $x < i$. There are thus at most three sets $S_j$, $j \neq i$, from which a site is in $\T_i'$.
     As each set $S_1,\dots,S_k$ contains at most $\lceil n/k \rceil$ sites, it follows that $|S_i' \cap S| \leq 4\lceil n/k \rceil = O(n/k)$. 
\end{proof}

\begin{restatable}{lemma}{treesCoverT}\label{lem:trees_cover_T}
    The trees $\T_i'$ cover $\T$, i.e. $\bigcup_i \T_i' = \T$.
\end{restatable}
\begin{proof}
    We prove the lemma by induction. Let $v$ be a vertex of $\T$ where a Steiner point is placed, then the induction hypothesis is that $\T(v)$ is covered by all subtrees $\T_i'$ corresponding to Steiner points in $\T(v)$. 
    
    As the base case, consider a vertex $v$ at which there is at least one Steiner point and for which there are no Steiner points in $\T(v)\setminus \{v\}$. If there is a single Steiner point $s_i$ at $v$, then $\T_i'= \T(v)$. Next, we consider the case that there is more than one Steiner point at $v$. Let $e$ be an edge incident to $v$ in $\T(v)$ and $x$ be the smallest color of $e$. If $s_x \neq v$ and $e$ has no additional color, then $\T(e)$ is included in $\T_j'$, where $s_j$ is the lowest or highest numbered Steiner point at $v$. If $e$ does have an additional color $y$, then $s_y = v$, as no site of $S_y$ can appear in $\T \setminus \T(v)$. It follows that $\T(e)$ is included in $\T_y'$. Finally, consider the case that $s_x = v$. Then $\T(e)$ is either included in $\T_x'$, or in a tree $\T_y'$, where $y$ is the potential other color of $e$. We conclude that all subtrees of $v$ are included in some $\T_i'$, and thus that in the base case the hypothesis holds.

    Let $v$ be a vertex of $\T$ that has a Steiner point and for which there is at least one Steiner point in $\T(v) \setminus \{v\}$. We assume the induction hypothesis holds for all Steiner points in $\T(v) \setminus \{v\}$. The exact same argument as for the base case tells us that each edge $e$ incident to $v$ in $\T(v)$ is included in some $\T_i'$ for a Steiner point $s_i$ at $v$. By definition, these subtrees include all of $\T(v)$ up to the next Steiner points. As the induction hypothesis holds for all Steiner points in $\T(v)$ below $v$, we conclude that $\T(v)$ is covered by the subtrees.
    
    Finally, we need to show there is a Steiner point at the root~$r$ of $\T$, as this would imply $\T(v)$ is covered. Because there are no uncolored edges or vertices after the coloring is complete, root~$r$ has at least one color $i$. This implies that there is a Steiner point $s_i$ at~$r$.
\end{proof}

\begin{restatable}{lemma}{treesDisjoint}\label{lem:trees_disjoint}
    The trees $\T_i'$ are disjoint apart from Steiner points.
\end{restatable}
\begin{proof}
    Suppose there is a non-Steiner vertex $v \in \T$ for which $v \in \T_i'$ and $v \in \T_j'$, $i < j$. The Steiner points $s_i$ and $s_j$ are both above $v$ in $\T$ by definition. If $s_i = s_j$, then the first edge $e$ of $\pi(s_i,v)$ has colors $i$ and $j$. However, as $i < j$, $\T(e)$ would not be included in $\T_i'$, which is a contradiction. If $s_i \neq s_j$, assume that $s_i$ is below $s_j$ in $\T$, the proof for $s_j$ below $s_i$ is similar. The path $\pi(v, s_j)$ must visit $s_i$, which contradicts that $v \in \T_j'$. The same argument proves that an edge $e \in \T$ can also be contained in at most one subtree $\T_i'$.
\end{proof}

\begin{restatable}{lemma}{fiveSteinerPoints}\label{lem:five_Steiner}
    There are at most five Steiner points in $\T_i'$.
\end{restatable}
\begin{proof}
    By definition $s_i$ is in $\T_i'$, so we want to show that there are at most four other Steiner points in $\T_i'$.  Note that a Steiner point can occur in $\T_i'$ only if its path to $s_i$ does not encounter any other Steiner point. We first consider the number of Steiner points in subtrees $\T(e)$ for which the edge $e$ is an edge incident to $s_i$ in $\T_i'$ that does not have color $i$, and then consider the number of Steiner points in subtrees for which $e$ does have color $i$.

    Let $\T_i'[j]$ be the subtree
    of $\T_i'$ rooted at $s_i$ that is the union of $\T(e) \cap \T_i'$ for all edges $e$ incident to $s_i$ of color $j \neq i$ and not of color $i$ as well, see Figure~\ref{fig:five_steiner}(a). 
    We argue that this subtree is non-empty for at most two colors $j$. 
    Consider such an edge $e$. 
    Because $e$ does not have color~$i$ and~$e \in \T_i'$, it must be that $s_j$ is above $s_i$ in $\T$. 
    Thus, the parent edge of $s_i$ in $\T$ must also be colored $j$. 
    By Lemma~\ref{lem:two_colors}, the parent edge of $s_i$ in $\T$ can be assigned at most two colors, so $\T_i'[j]$ is non-empty for at most two colors. 

         \begin{figure}
        \centering
        \includegraphics[page=4]{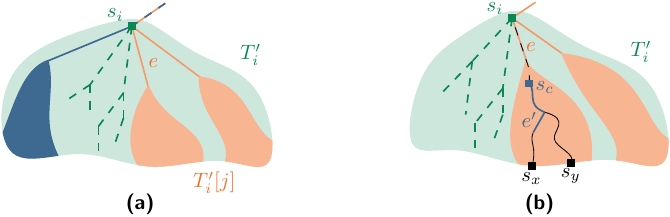}
        \caption{Notation used in Lemma~\ref{lem:five_Steiner}. In $\T_i'[j]$ are all subtrees that start with an edge of color~$j$.}
        \label{fig:five_steiner}
    \end{figure}
    
    Next, we prove that $\T_i'[j]$ contains at most one Steiner point other than $s_i$. We assume that $i < j$, the proof for $i > j$ is symmetric. Assume for contradiction that $\T_i'[j]$ contains two Steiner points $s_x$ and $s_y$, $x < y$; see Figure~\ref{fig:five_steiner}(b). As shown before, there is a site of $S_j$ in $\T \setminus \T(s_i)$.  As $i < j$, this implies that $i < x < y < j$. Let $e'$ be the first edge on $\pi(s_i,s_x)$ that is not on $\pi(s_i,s_y)$, i.e. the first edge after the paths diverge. Let $c$ be a color of $e'$ and let $v$ and $w$ be the upper and lower endpoint of $e'$. The tree $\T(w)$ does not contain any sites of $S_j$, as these appear in the traversal after the sites of $S_x$. It follows that $S_c$ is before $S_j$ in the in-order traversal, in other words $i < c < j$. The parent edge of $s_i$ cannot be colored $c$, as a site of $S_c$ would then appear either before a site in $S_i$ or after a site in $S_j$ in the in-order traversal. It follows that $s_c$ is on the path $\pi(v,s_i)$. If $s_c \neq s_i$, this contradicts the assumption that this path does not contain a Steiner point. If $s_c = s_i$, then $i < c$ implies that the subtree starting with an edge of color $j$ is in not $\T_i'$, which is a contradiction. We conclude that there are at most two Steiner points in the subtrees $\T_i'[j]$ in total for all $j \neq i$.
        
    What remains is to show that the subtrees whose edge incident to $s_i$ has color $i$ contain at most two Steiner points. The proof is similar to the proof for $\T_i'[j]$. We assume for contradiction that there are three such Steiner points $s_x$, $s_y$, $s_z$, with $x<y<z$. Assume that $\pi(s_i,s_y)$ diverges later from $\pi(s_i,s_x)$ than from  $\pi(s_i,s_z)$. The other case follows by symmetry. Consider the first edge $e'$ on $\pi(s_i,s_y)$ after which the path to $s_x$ diverges. Let $c$ be a color of $e'$. The in-order traversal implies that $x < c < z$. Because there is no Steiner point on the path from $s_y$ to $s_i$, the Steiner point $s_c$ must lie at or above $s_i$. However, $s_c$ cannot be above $s_i$ as this would imply that a site of $S_c$ appears before a site of $S_x$ or after a site of $S_z$ in the traversal, which is a contradiction. It follows that $s_i = s_c$.
    Consider the first edge $e$ on $\pi(s_i,s_y)$ and let $v$ be its lower endpoint. This edge has colors $i$ and $c$.
    So, there must be a site $p \in S_c$ in $\T \setminus \T(v)$. 
    When both $s_x$ and $s_z$ are in $\T(v)$, then $p$ must appear before a site in $S_x$ or after a site in $S_z$ in the traversal, which is a contradiction. If $s_z \notin \T(v)$ and $s_x \in \T(v)$, then $x < c$ again implies that $p$ must appear after the sites in $\T(v)$ in the traversal. Also, there must be a site $q \in S_i$ in the subtree that contains $s_z$, as by definition the first edge of this subtree has color $i$. This leads to a contradiction, as $p$ appearing after $q$ in the traversal would imply $p$ appears between two sites in $S_c$, and $p$ appearing before $q$ would imply $q$ appears between two sites in $S_i$. Finally, consider the case where $s_x$ and $s_z$ are both not in $T(v)$. As both of the subtrees that contain $s_x$ and $s_z$ start with an edge of color $i$, there must be a site of $S_i$ in both subtrees. However, this implies that one of these sites appears before the sites in $S_y$ and the other after, which is again a contradiction. So, there are at most two Steiner points in all subtrees whose first edge has color $i$.

    We conclude that there are at most five Steiner points in $\T_i'$.
\end{proof}

We are now ready to prove that our algorithm computes a spanner with low complexity.

\begin{lemma}\label{lem:steiner_spanner_tree}
    The spanner $\G$ is a $2t$-spanner for $\T$ of size $O(n\log (n/k))$ and complexity $O(mn^{1/t}/k^{1/t} + n\log (n/k))$.
\end{lemma}
\begin{proof}
    To bound the size and complexity of the spanner, we first consider the number of leaves~$n_i$ and vertices~$m_i$ in each subtree $\T_i'$. As $n_i$ is equal to $|S_i'|$ plus the number of Steiner points in $\T_i'$, Lemmata~\ref{lem:size_S_i} and~\ref{lem:five_Steiner} imply that $n_i = O(n/k) + 5 = O(n/k)$. 
    Lemma~\ref{lem:trees_disjoint} states the subtrees $\T_i'$ are disjoint apart from their shared Steiner points, so $\sum m_i = O(m)$. 
    By Lemma~\ref{lem:t-spanner-tree}, $\G$ has size $\sum_{i=1}^k O\left(\frac{n}{k} \log \left(\frac{n}{k}\right)\right)  = O\left(n \log \left(\frac{n}{k}\right)\right)$ and complexity $\sum_{i=1}^k O\left(m_i \left(\frac{n}{k}\right)^{1/t} + \frac{n}{k} \log\left(\frac{n}{k}\right)\right)
        = O\left(\frac{mn^{1/t}}{k^{1/t}} + n\log \left(\frac{n}{k}\right)\right)$.
    
    What remains is to show that $\G$ is a $2t$-spanner. Let~$p,q \in S$ be two leaves in~$\T$. If~$p,q \in S_i'$ for some $i \in \{1,\dots, k\}$ then the shortest path $\pi(p,q)$ is contained within~$\T_i'$. The $2t$-spanner on $\T_i'$ implies that $d_\G(p,q)\leq 2td(p,q)$. 
    If there is no such set $S_i'$ that contains both sites, 
    consider the sequence of vertices $v_1,\dots,v_\ell$ where $\pi(p,q)$ exits some subtree $\T_i'$. Let $v,w$ be two consecutive vertices in this sequence. Without loss of generality, assume that $w \in \T(v)$, and let $s_x$ be the Steiner point at $v$ for which $w \in \T_x'$ (Lemmata~\ref{lem:trees_cover_T} and~\ref{lem:trees_disjoint} imply $s_x$ exists). 
    Then the $2t$-spanner on $\T_x'$ ensures that $d_\G(v,w) \leq 2t d(v,w)$. It follows that $d_\G(p,q) \leq d_\G(p,v_1) + d_\G(v_1,v_2) + \dots + d_\G(v_\ell, q) \leq 2t(d(p,v_1) + d(v_1,v_2) + \dots + d(v_\ell,q)) = 2td(p,q)$.    
\end{proof}

\begin{restatable}{theorem}{steinerSpannerTree}\label{thm:steiner_spanner_tree}
    Let $\T$ be a tree with $n$ leaves and $m$ vertices, and $t \leq 1$ be any integer constant. For any $1 \leq k\leq n$, we can build a $2t$-spanner~$\G$ for~$\T$ using at most $k$ Steiner points of size $O(n\log (n/k))$ and complexity $O(mn^{1/t}/k^{1/t} + n\log (n/k))$ in $O(n\log(n/k) + m + K)$ time, where $K$ is the output size.
\end{restatable}
\begin{proof}
    Lemma~\ref{lem:steiner_spanner_tree} proves the spanning ratio, size, and complexity bound of $\G$. What remains is to show that $\G$ can be constructed in $O(n\log(n/k) + m + K)$ time.
    
    The construction of~$\G$ consists of three parts, partitioning the sites into sets~$S_1,\ldots,S_k$ to place the Steiner points, obtaining the trees~$\T_1',\ldots,\T_k'$, and computing a $2t$-spanner for each~$\T_i'$ using Lemma~\ref{lem:t-spanner-tree}. The first part of the construction starts with an in-order traversal of~$\T$ to obtain the sets~$S_1,\ldots,S_k$ in $O(m)$ time. 
    Since the sites are leaves of~$\T$, we can use leaf-to-root paths, to first find the smallest subtree that contains all sites of~$S_i$ for each color~$i$: Find the root~$r$ of the smallest subtree for~$S_i$ as the lowest vertex shared by the leaf-to-root paths from the first and last site encountered in the in-order treewalk for each~$S_i$, and color every vertex and edge in the subtree rooted at~$r$. Since every edge can get at most two colors, this takes~$O(m)$ time.
    Then proceed in bottom-up fashion to color the remaining uncolored edges and vertices. Once a color~$j$ is no longer advanced upwards from a vertex~$v$, that is, when another color is smaller and is used to color the edge upwards from~$v$ in the bottom-up approach, we can set the Steiner point~$s_j$ at vertex~$v$. Every edge of~$\T$ is considered only once in this phase of the process, and hence it takes~$O(m)$ time. 
    
    We find the assignment of vertices and edges to the trees~$\T_1',\ldots,\T_k'$ with another in-order treewalk of~$\T$, which makes use of a stack to remember the index of the tree we are currently assigning to. At the root of $T$ there must be a Steiner point~$s_i$ so start the in-order traversal with current index~$I=i$. Assign every vertex and edge we encounter on the walk to tree $\T_I'$, and whenever a vertex with Steiner points is encountered, push the index~$I$ on the stack and set~$I=j$ to be the index of the encountered Steiner point~$s_j$. Now proceed the treewalk, assigning edges and vertices to~$\T_I'$. When the in-order traversal returns to~$s_j$ after having walked over all descendant edges, we have completed~$\T_j'$ and returned to the vertex that is a leaf of~$T_i'$. We can pop the topmost index off the stack, which will be~$i$, and continue assigning to~$T_i'$.
    
    In case we encounter a vertex~$v$ with multiple Steiner points for the first time, push the current index~$I$ on the stack and set $I=g$ corresponding to the index of the Steiner point at this vertex with the lowest index. Once the in-order treewalk returns to~$v$ and has to continue walking on a descendant edge with color~$h>g$, we have completed~$T_g'$ and can safely set~$I$ to the color of the next lowest indexed Steiner point at~$v$. Once assigned to~$I$, the highest index of the Steiner points at such a vertex~$v$ is used until the treewalk has walked over all descendant edges: All trees $T_x'$ rooted at~$v$ have been completed, and the index of the tree which had~$v$ as a leaf can be popped off the stack. 
    Once the in-order traversal is finished, every vertex and edge has been assigned to a tree~$\T_x'$ according to the index~$I=x$ that was maintained at that point in the traversal. Since an index is pushed to the stack only when the root of another tree~$\T_y'$ is encountered, and every push is matched to a single pop operation, when going into and coming out of a subtree at said root, we use $O(k)$ constant time interactions with the stack. This procedure hence takes~$O(m+k)=O(m)$
    time.

    Finally, we use Lemma~\ref{lem:t-spanner-tree} to assess the running time for the computation of the $2t$-spanners on~$\T_1',\ldots,\T_k'$. Recall that Lemmata~\ref{lem:size_S_i},~\ref{lem:trees_disjoint}, and~\ref{lem:five_Steiner} imply for each $\T_i'$ that $n_i = O(n/k)$ and that $\sum(m_i) = m$. The total running time to compute the $2t$-spanners is then
    \begin{equation*}
        \sum_{i=1}^k O\left(\frac{n}{k} \log \frac{n}{k} + m_i\right) = k \cdot O\left(\frac{n}{k} \log \left(\frac{n}{k}\right)\right) + O(m) = O\left(n \log \left(\frac{n}{k}\right) + m\right). 
    \end{equation*}

    Thus, computing the spanners asymptotically dominates the total running time.
\end{proof}

\subsection{Extending the tree spanner to a forest}\label{sub:forest_spanner}
In this section, we extend our tree spanner to a spanner for a forest $\F$. We denote the trees in $\F$ by $\for_1,\dots,\for_\ell$ and by $S(\for_i)$ the sites in $\for_i$. As $\F$ is disconnected, we cannot require all sites in $S$ to have a path between them in the spanner. Instead, we say that $\G$ is a $t$-spanner for $\F$ if $\G$ is a $t$-spanner for every $\T_i$.

Let $k'= \lfloor k/2 \rfloor$. To construct a spanner on $\F$, we partition the sites into $k'$ sets $S_1,\dots S_{k'}$ using an in-order traversal of $\F$. In particular, we perform an in-order traversal of $\for_1$, then continue with an in-order traversal of $\for_2$, etc. As before, we assign the first $\lceil n/k' \rceil$ sites we encounter to $S_1$, the second $\lceil n/k' \rceil$ sites to $S_2$, and so on. A set $S_i$ can be distributed over many trees, however, only the first and last tree that has sites in $S_i$ can also contain sites from another set $S_j$. For a tree $\for$ whose sites are all assigned to the same set, i.e. $S(\for) \subseteq S_i$ for some $i \in \{1,\dots,k'\}$, we construct a spanner on $\for$ using the algorithm of Lemma~\ref{lem:t-spanner-tree} and add these edges to $\G$. These trees thus do not contain any Steiner points. For a tree $\for$ whose sites are not all assigned to the same set, i.e. $S(\for) \nsubseteq S_i$ for all $i \in \{1,\dots,k'\}$, let $K$ be the number of sets $S_i$ for which $S(\for) \cap S_i \neq \emptyset$.
We construct the spanner of Theorem~\ref{thm:steiner_spanner_tree} on $\for$ using $K$ Steiner points
and add the edges to $\G$. Because the intersection $S_i \cap \for(S)$ is non-empty for a set $S_i$ in at most two trees, the total number of Steiner points we place is at most $2k' \leq k$.

\begin{restatable}{theorem}{spannerForest}\label{thm:steiner_spanner_forest}
    Let $\F$ be a forest with $n$ leaves and $m$ vertices, and $t \leq 1$ be any integer constant. For any $1 \leq k \leq n$, we can build a $2t$-spanner $\G$ for $\F$ using at most $k$ Steiner points of size $O(n\log (n/k))$ and complexity $O(mn^{1/t}/k^{1/t} + n\log (n/k))$ in $O(n\log(n/k) + m + K)$ time, where $K$ is the output size.
\end{restatable}
\begin{proof}
    Clearly, the spanner~$\G$ constructed by the above procedure is a $2t$-spanner for $\for_1, \dots ,\for_\ell$ and thus a $2t$-spanner for $\F$. To bound the size and complexity of $\G$, and the running time of the algorithm constructing~$\G$, we first consider the trees whose sites are all in a single set $S_i$, and then consider the other trees. Let $\Min$ and $\Mout$ denote the total number of vertices in all trees whose sites are in a single set, and the total number of vertices in the other trees, respectively. Thus $m = \Min + \Mout$.

    Let $a_i$ denote the number of trees whose sites are fully contained within the set $S_i$. Let $n_{i,j}$ and $m_{i,j}$ denote the number of sites and vertices in the $j$-th subtree whose sites are contained in $S_i$. Note that $\sum_j n_{i,j} \leq  \lceil n/k'\rceil$ and $\sum_{i,j} m_{i,j} = \Min$. We obtain the following bound on the number of links of the spanner restricted to these trees by Lemma~\ref{lem:t-spanner-tree}
    \begin{equation}\label{eq:size_in}
        \sum_{i =1}^{k'} \sum_{j=1}^{a_i} n_{i,j} \log n_{i,j} \leq \sum_{i =1}^{k'}\sum_{j=1}^{a_i} n_{i,j} \log \left\lceil \frac{n}{k'}\right\rceil \leq \log \left\lceil \frac{n}{k'}\right\rceil  \sum_{i =1}^{k'}\left\lceil \frac{n}{k'}\right\rceil = O\left(n \log \left(\frac{n}{k}\right)\right).
    \end{equation}
    Similarly, Lemma~\ref{lem:t-spanner-tree} bounds the complexity of $\G$ restricted to these trees as follows
    \begin{equation}\label{eq:complexity_in}
    \begin{split}
        \sum_{i =1}^{k'} \sum_{j=1}^{a_i} (m_{i,j} n_{i,j}^{1/t} + n_{i,j} \log n_{i,j}) &\leq \left\lceil \frac{n}{k'} \right\rceil^{1/t} \sum_{i =1}^{k'} \sum_{j=1}^{a_i} m_{i,j} + O\left(n \log \left(\frac{n}{k}\right)\right) \\
        &= O\left(\frac{\Min n^{1/t}}{k^{1/t}} + n \log \left(\frac{n}{k}\right)\right).
    \end{split}
    \end{equation}
    Lastly, the algorithm of Lemma~\ref{lem:t-spanner-tree} is used to construct the $2t$-spanners on these trees, which results in a total running time of
    \begin{equation}\label{eq:runtime_in}
        \sum_{i =1}^{k'} \sum_{j=1}^{a_i} n_{i,j} \log n_{i,j} + m_{i,j}\leq \Min + \sum_{i =1}^{k'}\sum_{j=1}^{a_i} n_{i,j} \log \left\lceil \frac{n}{k'}\right\rceil \leq O\left(n \log \left(\frac{n}{k}\right) + \Min\right).
    \end{equation}
    
    Next, we consider the trees whose sites are not contained within a single set. Let $\for_j$ be such a tree. Furthermore, let $n_j$, $m_j$, and $k_j$ denote the number of sites, vertices, and Steiner points in $\for_j$. 
    As $k_j$ is equal to the number of sets for which $S_i \cap \for_j \neq \emptyset$, we have $n_i \leq k_j \lceil n/k'\rceil$.
    Theorem~\ref{thm:steiner_spanner_tree} states that the size of $\G$ restricted to $\for_j$ is $O\left(\frac{k_jn}{k'}\log \left( \frac{n}{k'}\right)\right)$, the complexity is $O\left(\frac{m_jn^{1/t}}{k'^{1/t}} + \frac{k_jn}{k'}\log \left( \frac{n}{k'}\right)\right)$, and the time to compute this part of~$\G$ is $O\left(\frac{k_jn}{k'}\log \left( \frac{n}{k'}\right) + m_j\right)$.
    
    Recall that $\sum_j k_j \leq k$ and $\sum_j m_j = \Mout$. Summing over all trees $\for_j$ whose sites are not contained within a single set, we obtain the following expressions for the total size and complexity of their spanners, and the running time for constructing these $2t$-spanners.
    \begin{equation}\label{eq:size_out}
        \sum_j O\left(\frac{k_jn}{k'}\log \left( \frac{n}{k'}\right)\right) = O\left(\frac{k}{k'} \cdot n \log \left(\frac{n}{k'} \right)\right) = O\left(n \log \left(\frac{n}{k}\right)\right).
    \end{equation}
    \begin{equation}\label{eq:complexity_out}
        \sum_j O\left(\frac{m_jn^{1/t}}{k'^{1/t}} + \frac{k_jn}{k'}\log \left( \frac{n}{k'}\right)\right) = O\left(\frac{\Mout n^{1/t}}{k^{1/t}} +  n \log \left(\frac{n}{k}\right) \right).
    \end{equation}
    \begin{equation}\label{eq:runtime_out}
        \sum_j O\left(\frac{k_jn}{k'}\log \left( \frac{n}{k'}\right) + m_j\right) = O\left(\frac{k}{k'} \cdot n \log \left(\frac{n}{k'} \right) + \Mout\right) = O\left(n \log \left(\frac{n}{k}\right) + \Mout\right).
    \end{equation}

    Summing Equations~(\ref{eq:size_in}) and~(\ref{eq:size_out}) we obtain the stated bound on the size of $\G$. By summing Equations~(\ref{eq:complexity_in}) and~(\ref{eq:complexity_out}), and noting that $m = \Min + \Mout$, we obtain the complexity of $\G$. To obtain total running time for constructing~$\G$, recall that the algorithm consists of two steps: First, an in-order traversal of~$\F$, which takes~$O(m)$ time. Second, constructing the $2t$-spanners for the trees of~$\F$, which takes time equal to the sum of Equations~(\ref{eq:runtime_in}) and~(\ref{eq:runtime_out}).
\end{proof}

\section{Steiner spanners in simple polygons} 
\label{sec:spanners_in_a_simple_polygon}
We consider the problem of computing a $t$-spanner using $k$ Steiner points for a set of $n$ point sites in a simple polygon $P$ with $m$ vertices. 
We measure the distance between two points in $p,q$ in $P$ by their geodesic distance, i.e. the length of the shortest path $\pi(p,q)$ fully contained within~$P$. A link $(p,q)$ in the spanner is the shortest path $\pi(p,q)$, 
and its complexity is the number of segments in this path.
Lower bounds for trees straightforwardly extend to polygonal instances.
Again, we aim to obtain a spanner of complexity close~to~the~lower~bound.

\begin{restatable}{lemma}{lowerBoundsPolygon}
    The lower bounds of Lemmata~\ref{lem:lowerbound_2spanner}, \ref{lem:lowerbound_3spanner}, and~\ref{lem:lowerbound_tspanner} also hold for simple polygons.
\end{restatable}
\begin{proof}
    As the trees used for the lower bounds are planar, we can construct a simple polygon~$P$ that has the tree as free space and place the sites at the locations of the leaves of the tree. The complexity of paths and distances between points in $P$ are the same as in the tree.
\end{proof}

\subparagraph{Spanner construction.} Next, we describe how to obtain a
low-complexity spanner in a simple polygon using at most~$k$ Steiner
points. In our approach, we combine ideas
from~\cite{SpannerPolyhedralTerrain} and~\cite{complexity_spanners}
with the forest spanner of Theorem~\ref{thm:steiner_spanner_forest}. We first
give a short overview of the approach to obtain a low complexity $2\sqrt{2}t$-spanner~\cite{complexity_spanners},
and then discuss how to combine these ideas with the forest spanner to
obtain a low complexity Steiner spanner.

We partition the polygon $P$ into two subpolygons $P_\ell$ and $P_r$ by a vertical line segment~$\lambda$ such that roughly half of the sites lie in either subpolygon. For the line segment $\lambda$, we then consider the following weighted 1-dimensional space. For each site $p \in S$, let $p_\lambda$ be the \emph{projection} of $p$: the closest point on $\lambda$ to $p$. The (weighted 1-dimensional distance) between two sites~$p_\lambda,q_\lambda$ is defined as $d_w(p_\lambda,q_\lambda) := d(p,p_\lambda) + d (p_\lambda,q_\lambda) + d(q,q_\lambda)$. In other words, the sites in the 1-dimensional space are weighted by the distance to their original site in~$P$. For this 1-dimensional space we construct a $t$-spanner $\G_\lambda$, and for each link $(p_\lambda,q_\lambda)$ in $\G_\lambda$ we add the link $(p,q)$ to the spanner~$\G$. Finally, we process the subpolygons $P_\ell$ and $P_r$ recursively. De Berg, van Kreveld, and Staals~\cite{complexity_spanners} show that this gives a $\sqrt{2}t$-spanner in a simple polygon. To obtain a spanner of complexity~$O(mn^{1/t} + n\log^2n)$, they construct a 1-dimensional $2t$-spanner $\G_\lambda$ using the approach of Lemma~\ref{lem:t-spanner-tree}, resulting in a $2\sqrt{2}t$-spanner.

In our case, we require information on the paths from the sites to their projection instead of only their distance to decide where to place the Steiner points. This information is captured in the shortest path tree $\SPT_\lambda$ of the segment~$\lambda$, which is the union of all shortest paths from the vertices of $P$ to their closest point on $\lambda$. Additionally, we include all sites in $S$ in the tree $\SPT_\lambda$. The segment $\lambda$ is split into multiple edges at the projections of the sites, see Figure~\ref{fig:shortest_path_tree}. 
The tree $\SPT_\lambda$ is rooted at the lower endpoint of $\lambda$ and has $O(m + n)$ vertices. 

\begin{figure}
    \centering
    \includegraphics{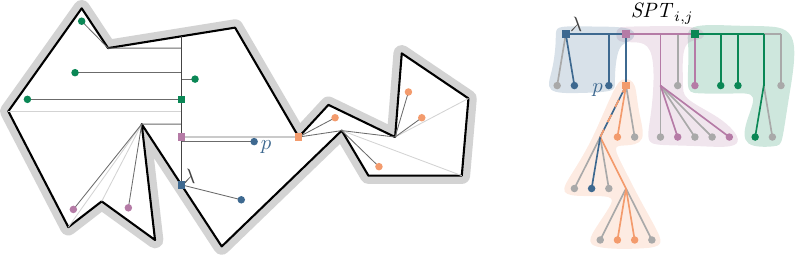}
    \caption{The shortest path tree of $\lambda$ in $P'$ and its $\SPT_{i,j}$. The grey nodes and edges are not included in $\SPT_{i,j}$, but can be assigned to a $\T_i'$ as indicated by the colored backgrounds. The squares show the Steiner points in $\SPT_{i,j}$ and $P'$. The sites in $P'$ are colored as the trees $\T_i'$. }
    \label{fig:shortest_path_tree}
\end{figure}

We adapt the algorithm to build a spanner in $P$ as follows. Instead of computing a 1-dimensional spanner directly in each subproblem in the recursion, we first collect the shortest path trees of all subproblems. Let $\SPT_{i,j}$ denote the shortest path tree of the $j$-th subproblem at the $i$-th level of the recursion. We exclude all vertices from $\SPT_{i,j}$ that have no site as a descendant. This ensures that all leaves of the tree are sites. Let $\F = \cup_{i,j} \SPT_{i,j}$ be the forest consisting of all trees. A site in $S$ or vertex of $P$ can occur in multiple trees $\SPT_{i,j}$, but they are seen as distinct sites and vertices in the forest $\F$. We call a tree $\SPT_{i,j}$ \emph{large} if $0 \leq i \leq \log k$ and \emph{small} otherwise. In other words, the trees created in the recursion up to level $\log k$ are large.
We then partition $\F$ into two forests $\F_s$ and $\F_\ell$ containing the small and large trees. For each tree in $\F_s$ we directly apply the $2t$-spanner of Lemma~\ref{lem:t-spanner-tree} that uses no Steiner points to obtain a spanner $\G_s$. For the forest $\F_\ell$ we apply Theorem~\ref{thm:steiner_spanner_forest} to obtain a $2t$-spanner $\G_\ell$ for $\F_\ell$. Let $\G_\F = \G_s \cup \G_\ell$. A Steiner point in $\G_\F$ corresponds to either a vertex of $P$ or a point on $\lambda$. Let $\steiner$ denote the set of Steiner points. To obtain a spanner~$\G$ in the simple polygon, we add a link $(p,q)$, $p,q \in S \cup \steiner$, to~$\G$ 
whenever there is a link in $\G_\F$ between (a copy of) $p$ and $q$.

\begin{restatable}{lemma}{simplePolygonSize}\label{lem:spanning_ratio_simple_polygon}
    The graph~$\G$ is a $2\sqrt{2}t$-spanner for the sites $S$ in $P$ of size $O(n\log ^2n)$.
\end{restatable}
\begin{proof}
    Because the spanner $\G_\F$ is also a spanner for the weighted 1-dimensional space, Lemma 2 of~\cite{complexity_spanners} directly implies that $\G$ is a $2\sqrt{2}t$-spanner. The size of $\G$ is the same as the size of $\G_\F$. Each site in $S$ occurs in one $\SPT_{i,j}$ for each level $i$. There are thus $O(n\log k)$ sites in $\F_\ell$ and $O(n(\log n-\log k))$ sites in $\F_s$. 
    Theorem~\ref{thm:steiner_spanner_forest} then implies that the size of $\G_{\ell}$ is $O(n\log k \log(n\log k/k)) = O(n\log^2 n)$.
    Lemma~\ref{lem:t-spanner-tree} implies that the size of $\G_{s}$ is $
        \sum_{i= \log k}^{O(\log n)} 2^i \cdot O\left(n/2^i \log (n/2^i)\right) = O(n\log^2n)$.
     It follows that the size of $\G$ is $O(n\log^2n)$.
\end{proof}

\subparagraph{Complexity analysis.}
To bound the complexity of the links in $\G$, we have to account for the complexity of links generated by both $\G_s$ and $\G_\ell$. Bounding the complexity of $\G_s$ is relatively straightforward, but to bound the complexity of $\G_\ell$ we first prove a lemma on the structure of a shortest path in $P$ between sites in $\F_\ell$.

Let $\T$ be a tree in $\F_\ell$ and let $P'$ be the corresponding subpolygon of the subproblem. We consider the shortest path between two sites that are assigned to the same subtree $\T_i'$ of~$\T$ by the forest algorithm. It can be that this shortest path uses vertices of $P'$ that were excluded from $\T$, as they had no site as a descendant. For the analysis, we do include these vertices in~$\T$ and assign them to subtrees $\T_j'$ as in Section~\ref{sub:tree_spanner}; see Figure~\ref{fig:shortest_path_tree}. The following lemma states that the complexity of a shortest path between two sites in the same subtree~$\T_i'$ is bounded by the number of vertices in~$\T_i'$. We use this to bound the complexity in Lemma~\ref{lem:complexity_spanner_simple_polygon}.

\begin{lemma}\label{lem:shortest_path_in_tree}
    A shortest path $\pi(p,q)$ in $P'$ between sites $p,q \in \T_i'$ uses vertices in~$\T_i'$ only.
\end{lemma}
\begin{proof}
    Assume for contradiction that $r$ is a highest vertex in $\T$ used by $\pi(p,q)$ that is not in~$\T_i'$. First, consider the case that $r$ is the root of $\T$. Recall that this means $r$ is the bottom endpoint of~$\lambda$, and thus lies on the boundary of $P'$. As $r$ is on $\pi(p,q)$, it must be that $p$ and~$q$ lie in different subpolygons, and at least one of them lies below the horizontal line through $r$. This implies that $s_i = r$, which is a contradiction.
    
    Next, consider the case that $r$ is not the root of $\T$. Let $r'$ be the parent of $r$. If $r'$ is in~$\T_i'$, then it must be a leaf. We consider the following partition of $P'$. Recall that $r_\lambda$ denotes the closest point on $\lambda$ to $r$. We extend the shortest path $\pi(r,r_\lambda)$ to the boundary of $P'$ by extending the first and last line segments of the path to obtain a path $\pi_r$, see Figure~\ref{fig:definition_Pwithr}(a). Let $\partial P'$ denote the boundary of $P'$. We define $\Pwithr$ to be the closed polygon bounded by $\partial P'$ and $\pi_r$ that contains the polygon edges incident to $r$, and $\Pnotr := P' \setminus \Pwithr$. Because $r$ is a reflex vertex of $P'$,  $\Pwithr$ is well-defined. Without loss of generality, we assume that $\Pwithr$ contains the part of $\lambda$ above $r_\lambda$, as in Figure~\ref{fig:definition_Pwithr}(a). If both $p$ and $q$ are in $\Pnotr$, then $r \notin \pi(p,q)$. It follows that $p$ and/or $q$ are in $\Pwithr$. Without loss of generality, assume that~$p \in \Pwithr$. 
    
    We distinguish two cases based on the location of $p$, see Figure~\ref{fig:definition_Pwithr}(a). Either $p \in A$, where $A \subset \Pwithr$ is bounded by the extension segment starting at $r$ and $\partial P'$, or $p \in B$, where $B \subset \Pwithr$ is bounded by $\pi(r,r_\lambda)$, the extension segment starting at~$r_\lambda$, and~$\partial P'$.

    \begin{figure}
        \centering
        \includegraphics[page=3]{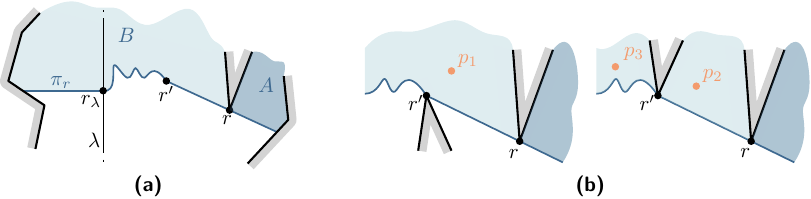}
        \caption{\textbf{\textsf{(a)}} The extended path $\pi_r$ separates the polygon into $\Pwithr = A \cup B$ and $\Pnotr$. \textbf{\textsf{(b)}}~Sites $p_{1\dots 3}$ correspond to the respective subcases (i--iii) based on the structure of the polygon around~$r'$.}
        \label{fig:definition_Pwithr}
    \end{figure}
    
    If $p \in A$, then $p$ is a descendant of $r$ in $\T$. As $p$ and $q$ are in $\T'_i$ and $r$ is not, it must be that $q$ is also a descendant of $r$. It follows that $q \in A$, but this means that $r$ is not a reflex vertex on $\pi(p,q)$, which contradicts it being a shortest path.

    If $p \in B$, the previous paragraph implies that $q \notin A$. Additionally, $q\notin B$ as well, as~$r$ would then not be a reflex vertex in $\pi(p,q)$. It follows that $q \in \Pnotr$. Next, we make a distinction on whether $r'$ is a vertex of $P'$ or not. First, assume that $r'$ is not a vertex of $P'$, and thus $r' \in \lambda$. Because $p \in B$, $p_\lambda$ must be at or above $r'$. Because $q \in \Pnotr$, $q_\lambda$ must be below $r'$. This implies that the path in $\T$ from $p$ to $q$ visits $r'$, which contradicts $p,q \in \T'_i$.

    Next, we assume that $r'$ is a vertex of $P'$. We distinguish three different subcases based on the shape of the polygon around $r'$, see Figure~\ref{fig:definition_Pwithr}(b), and find a contradiction in each case:

    \begin{description}[topsep=0pt]
        \item[(i) The edges of $P'$ incident to $r'$ are in $\Pnotr$.] As $r$ is on $\pi(p,q)$, $q$ must be a descendant of~$r'$. It follows that the Steiner point $s_i$ is located on the path in $\T$ from $q$ to $r'$, so $p$ is also a descendant of $r'$. It follows that $p$ is on the segment $rr'$. However, for $r$ to be on $\pi(p,q)$, $q$ must then be in $A$, which is a contradiction.
        \item[(ii) The edges of $P'$ incident to $r'$ are in $\Pwithr$, and $r'$ is on $\pi(p,p_\lambda)$.] In this case, $p$ is a descendant of~$r'$. This again implies that $q$ is a descendent of~$r'$, which contradicts~$q \in \Pnotr$.
        \item[(iii) The edges of $P'$ incident to $r'$ are in $\Pwithr$, and $r'$ is not on $\pi(p,p_\lambda)$.] The path $\pi(p,q)$ either intersects the boundary of $B$ twice, which is not allowed as both are shortest paths, or visits $r'$ as well. However, this implies that $q \in A$, which is a contradiction.\qedhere
    \end{description}
\end{proof}

\begin{lemma}\label{lem:complexity_spanner_simple_polygon}
    The spanner $\G$ has complexity $O(mn^{1/t}(\log k)^{1+1/t}/k^{1/t} + n\log^2 n)$.
\end{lemma}
\begin{proof}
    To bound the complexity of the links in $\G$ generated by 
    $\G_s$ we apply Lemma~\ref{lem:t-spanner-tree} directly. As Lemma~\ref{lem:t-spanner-tree} corresponds to the algorithm to construct a low complexity spanner in a polygon using the shortest path tree, the complexity bound also holds in the simple polygon setting. Using $\sum_{j=0}^{2^i} m_{i,j}= O(m)$, where $m_{i,j}$ is the number of vertices in $\SPT_{i,j}$, the complexity is
    \begin{align*}
        \sum_{i= \log k}^{O(\log n)} \sum_{j=0}^{2^i} O\left(m_{i,j}\left(\frac{n}{2^i}\right)^{1/t} + \frac{n}{2^i} \log \left(\frac{n}{2^i}\right)\right) 
        &= O\left(\frac{mn^{1/t}}{k^{1/t}} + n\log^2n\right).
    \end{align*}
    
    For $\F_\ell$, the algorithm of Lemma~\ref{lem:t-spanner-tree} is used as a subroutine on every subtree $\T_i'$. Lemma~\ref{lem:shortest_path_in_tree} implies that the complexity bound of Theorem~\ref{thm:steiner_spanner_forest} also holds for links in~$P$. Recall that the number of sites in $\F_\ell$ is $O(n\log k)$. A vertex of $P$ can occur in at most two subproblems at each level of the recursion that partitions $P$, thus the number of vertices in $\F_\ell$ is $O((m+n)\log k)$. As the $n$ sites are equally divided over all subproblems at level~$i$, the complexity of the links in $\G$ generated by $\G_\ell$ given by Theorem~\ref{thm:steiner_spanner_forest} is improved to 
    \begin{equation*}
        O\left(\frac{m\log k (n\log k)^{1/t}}{k^{1/t}} + n\log k \log \left(\frac{n\log k}{k}\right)\right)= O\left(\frac{mn^{1/t}(\log k)^{1+1/t}}{k^{1/t}} + n\log^2 n\right).\qedhere
    \end{equation*}
\end{proof}

\subparagraph{Running time.} 
We use a similar algorithm to de Berg, van Kreveld, and Staals~\cite{complexity_spanners} to efficiently compute our spanner. We briefly sketch the algorithm here. We first preprocess the polygon: We build a shortest path data structure~\cite{Chazelle_triangulate,2PSP_simple_polygon}, and the horizontal and vertical decomposition~\cite{Chazelle_triangulate,Kirkpatrick_point_location} in $O(m)$ time. Next, for each site we find the trapezoid (in the horizontal and vertical decomposition) that contains it. For each trapezoid, we sort the sites within on $x$-coordinate. This preprocessing takes $O(n(\log n + \log m))$ time.

Then, we recursively construct all shortest path trees $\SPT_{i,j}$. We find a vertical separator~$\lambda$ such that at most $2n/3$ sites lie in either subpolygon in~$O(n+m)$ time using the vertical decomposition (\cite[Lemma 10]{complexity_spanners}), and construct the shortest path tree $\SPT_\lambda$ in~$O(m + n \log m)$ time using the horizontal decomposition (\cite[Lemma 11]{complexity_spanners}). The time to construct the shortest path trees $\SPT_{i,j}$ for all $O(\log n)$ levels is thus $O(m\log n + n \log m \log n)$.

Finally, we construct the spanners on $\F_s$ and $\F_\ell$. Theorem~\ref{thm:steiner_spanner_tree} and Theorem~\ref{thm:steiner_spanner_forest} state that the running time to construct these spanners is asymptotic in their size plus the number of vertices in the forest. As there are $O((m+n) \log n)$ vertices in $\F$, the total running time is $O(n\log^2 n + m \log n)$, which dominates the shortest path tree construction time.

\begin{restatable}{theorem}{simplePolygonTheorem}\label{thm:simplepolygon}
    Let $S$ be a set of $n$ point sites in a simple polygon $P$ with $m$ vertices, and $t \geq 1$ be any integer constant. For any $1 \leq k\leq n$, we can build a geodesic $2\sqrt{2}t$-spanner with at most $k$ Steiner points, of size $O(n \log^2 n)$ and complexity $O(mn^{1/t}(\log k)^{1+1/t}/k^{1/t} + n\log^2 n)$ in $O(n\log^2 n + m \log n + K)$ time, where $K$ is the output size.
\end{restatable}

\subparagraph{A relaxed geodesic $(2k + \varepsilon)$-spanner.} In a more recent version of the paper by de Berg, van Kreveld, and Staals~\cite{full_version,complexity_spanners} they show how to apply the refinement proposed by Abam, de Berg, and Seraji~\cite{SpannerPolyhedralTerrain} to improve the spanning ratio to $(2k+\varepsilon)$ for any constant $\eps \in (0,2k)$. They make two changes in their approach. First, instead of using the shortest path between two sites as a link they allow a link to be any path between two sites. They call such a spanner a \emph{relaxed} geodesic spanner. Second, for each split of the polygon they construct spanners on several 
sets of sites in the 1-dimensional weighted space. Using the same adaptations, we obtain a relaxed $(2k+\eps)$-spanner of complexity $O(mn^{1/t}(\log k)^{1+1/t}/k^{1/t} + n\log^2 n)$.

\section{Steiner spanners in polygonal domains}
\label{sec:polygonal_domain}

If the polygon contains holes, the spanner construction in the previous section no longer suffices. In particular, we may need a different type of separator, and shortest paths in $P$ are no longer restricted to vertices in some subtree (Lemma~\ref{lem:shortest_path_in_tree} does not hold). De Berg, van Kreveld, and Staals~\cite{complexity_spanners} run into similar problems when generalizing their low complexity spanner, and solve them as follows. There are two main changes in their construction. First, the separator is no longer a line segment, but a \emph{balanced separator} that consists of at most three shortest paths that partition the domain into two subdomains $P_r$ and $P_\ell$. They then construct a spanner $\G_\lambda$ on the 1-dimensional space containing the projections of the sites for each shortest path in the separator. Second, the links that are included in the spanner are no longer shortest paths, but consist of at most three shortest paths, resulting in a relaxed geodesic spanner. In contrast to the simple polygon, using a 1-dimensional spanner with spanning ratio $t$ results in a spanning ratio in~$P$ of~$3t$~\cite{complexity_spanners}.

To construct a low complexity spanner using $k$ Steiner points, we use our simple polygon approach with the adaptions of~\cite{complexity_spanners}. The number of trees, and thus the number of sites and vertices in the trees, increases by a constant factor, as we create at most three shortest path trees at each level. To bound the complexity, we can no longer apply Lemma~\ref{lem:shortest_path_in_tree}. However, the links that are added to $\G$ are shortest paths in the shortest path tree. Therefore, the bound on the complexity of $\G_\F$ directly translates to a bound on the complexity of $\G$. As in the simple polygon case, we obtain a spanner of complexity $O(mn^{1/t}(\log k)^{1+1/t}/k^{1/t} + n\log^2 n)$.

\begin{restatable}{theorem}{polygonaldomain}
    Let $S$ be a set of $n$ point sites in a polygonal domain $P$ with $m$ vertices, and $t \geq 1$ be any integer constant. 
    For any $k \leq n$, we can build a relaxed geodesic $6t$-spanner with at most $k$ Steiner points, of size $O(n \log n \log (n/k))$ and complexity $O(mn^{1/t}(\log k)^{1+1/t}/k^{1/t} + n\log^2 n)$ in $O(n \log^2 n + m \log n \log m +  K)$ time, where $K$ is the output size.
\end{restatable}
\begin{proof}
The stated spanning ratio, size, and complexity bounds follow from Theorem~\ref{thm:simplepolygon} together with the adaptation described in Section~\ref{sec:polygonal_domain}. What remains is to prove the stated running time.
    De Berg, van Kreveld, and Staals~\cite{full_version,complexity_spanners} show that a balanced separator can be computed in $O(m \log m + n \log n)$ time.
After that, the shortest path tree $\SPT_\lambda$ of a shortest path~$\lambda$ in the separator can be constructed in $O((n+m)\log m)$ time~\cite{SPMHershbergerSuri}. The time to construct all balanced separators and the shortest path trees is thus $O(n \log^2 n + m \log n \log m)$. As in the simple polygon case, we can build the spanner on $\F$ in $O(n\log^2n + m\log n)$ time. So, the total running time is $O(n\log^2n + m \log n\log m + K)$.
\end{proof}

\section{NP-hardness}
\label{sub:np-hardness}

    Given a polygonal domain $P$ with $m$ vertices and a set $S$ of $n$ sites in $P$ as input, the \steinerspanner problem asks whether there exist a spanner with spanning ratio $t$ that uses at most~$k$ Steiner points and has complexity at most $C$. 
    In this section, we show that this problem is NP-hard already for $t=3$ using a reduction from \vertexcover in penny graphs, which is an NP-complete problem~\cite{vertexcover_penny}. 
    Recall that a \emph{vertex cover} of a graph $G=(V,E)$ is a set of vertices $W\subseteq V$ such that each edge in $E$ has at least one incident vertex in $W$. 
    The \vertexcover problem cover asks whether a given graph has a vertex cover of size at most~$k$.
    \emph{Penny graphs} are contact graphs of unit disks. 
    Equivalently, they are the graphs that admit a \emph{penny graph drawing}: a straight-line drawing in the plane in which no two edges cross, all edges are unit-length, and the angle between any two edges emanating from the same vertex is at least $\pi/3$. 
   
    To show that the \steinerspanner problem is NP-hard, we first show in Lemma~\ref{lem:without_degree_one} that \vertexcover remains hard on penny graphs without vertices of degree one.
    
    \begin{lemma}\label{lem:without_degree_one}
        \vertexcover restricted to penny graphs without degree-one vertices is NP-complete, even if we are provided with a corresponding penny graph drawing.
    \end{lemma}
    \begin{proof}
        Clearly the problem is in NP.
        We reduce from \vertexcover restricted to penny graphs, which has been shown to be NP-complete~\cite{vertexcover_penny}.
        Let $G = (V,E)$ be a penny graph.
        We construct a penny graph $G' = (V',E')$ as follows. 
        Repeatedly remove a degree one vertex and its neighbor until there are no degree-one vertices left.
        Let $R$ be the set of vertices removed from $G$, i.e. $R = V \setminus V'$. 
        We show that $G$ has a vertex cover $\VC$ of size at most $k$ if and only if $G'$ has a vertex cover $\VC'$ of size at most $k'$, where $k' = k - |R|/2$. 
        Consider a degree-1 vertex $u$ of $G$ and let $v$ be its only neighbor. 
        If there is a vertex cover of $G$ of size $k$ that includes $u$, we can produce a vertex cover of size at most $k$ that includes $v$ and not $u$. 
        Thus, there is always a solution for \vertexcover that includes $v$ and not $u$, and any vertex cover of $G-\{u,v\}$ can be extended to a vertex cover of $G$ by adding a single vertex $u$. 
        Iterating this argument yields the desired result. 
        We remark that the fact that we can iteratively remove degree-1 vertices in this way is a standard reduction rule for finding a kernel for \vertexcover~\cite{CyganFKLMPPS15parameterized}.
    \end{proof}
    
    To obtain a reduction from \vertexcover on penny graphs to \steinerspanner, we require that there are no small cycles in the graph. 
    Define an \emph{$s$-subdivided penny graph} to be a graph obtained from a penny graph by replacing each edge by a path of $s$ edges. 
    Observe that an $s$-subdivided penny graph has degree-one vertices if and only if the original penny graph has degree-one vertices.
    In Lemma~\ref{lem:split_edges} we show that \vertexcover remains hard on $3$-subdivided penny graphs without vertices of degree one.
    
    \begin{lemma}\label{lem:split_edges}
        \vertexcover restricted to $3$-subdivided penny graphs without degree-one vertices is NP-complete, even if we are provided with a corresponding penny graph drawing.
    \end{lemma}
    \begin{proof}
        Again, NP membership is trivial. We reduce from \vertexcover restricted to penny graphs without degree-one vertices, which is NP-complete by Lemma~\ref{lem:without_degree_one}. Let $G = (V,E)$ be a penny graph without degree-one vertices, and let $G' = (V', E')$ be the corresponding $3$-subdivided penny graph.
        Let $k$ be the size of a minimal vertex cover of $G$, and $k'$ be the size of a minimal vertex cover of $G'$.
        It suffices to show that $k'=k+|E|$.
        
        We can extend any vertex cover $\VC$ of $G$ to a vertex cover of $G'$ by adding a single vertex to the cover for each subdivided edge $e\in E$: $\VC$ contains at least one endpoint of $e$, so $\VC$ already covers the first or last edge of the path of three edges that replace $e$ in $G'$. The remaining two edges that replace $e$ can be covered by adding a single vertex to the cover.
        The resulting vertex cover of $G'$ has size $|\VC|+|E|$, so $k'\leq k+|E|$.
        
        For the reverse direction, we show that there exists a vertex cover $\VC'$ of $G'$, in which for each path $(u,u')-(u',v')-(v',v)$ that replaces an edge $(u,v)$ of $G$, exactly one of $u'$ and $v'$ lie in $\VC'$.
        Indeed, at least one must lie in $\VC'$ to cover the edge $(u',v')$, and if both lie in $\VC'$, we maintain a vertex cover by replacing $v'$ by $v$.
        It follows that $\VC'$ contains at least one of $u$ and $v$ for each edge $(u,v)$ of $G$, and hence $\VC'\cap V$ is a vertex cover for $G$ of size $|\VC'|-|E|$, so $k'\geq k+|E|$. 
        Therefore $k'=k+|E|$, which completes the proof.
    \end{proof}
    
    Next, we show that \steinerspanner is NP-hard by a reduction from \vertexcover on $3$-subdivided penny graphs without degree-one vertices. We construct a polygonal domain and a set of sites as follows, see Figure~\ref{fig:np_hard_construction} for an example. 

\subparagraph{Construction.} 
    Let $G=(V,E)$ be a $3$-subdivided penny graph without degree-one vertices.
    Based on a penny graph drawing of the corresponding `unsubdivided' graph, we obtain a penny graph drawing of $G$ in which all `subdivided edges' are straight, see Figure~\ref{fig:np_hard_construction} (left).
    We construct a polygonal domain whose free space consists of a small neighborhood around the vertices and edges of such a drawing.
    Near each vertex $v$ of $G$, our polygonal domain contains a \emph{corridor} that splits into two corridors of length $\eps<1/8$ 
    and place a site at the end of each corridor, see Figure~\ref{fig:np_hard_construction} (right). 
    Let $p_h(v)$ and $p_\ell(v)$ denote these sites.
    These corridors are `carved' into a hole adjacent to the vertex. 
    The corridor to $p_h(v)$ has a high complexity of $3M$, while the corridor of $p_\ell(v)$ has a lower complexity of~$M$, with $M > 26n$. 
    Additionally, we include a convex chain of $M$ vertices between the two corridors, such that the shortest path $\pi(p_h(v),p_\ell(v))$ has complexity $5M$, while the shortest path from $p_\ell(v)$ to the entrance has a complexity of just $M$. 
    We denote by $\delta<\eps/3$ the difference between the length of the shortest path $\pi(p_h(v),p_\ell(v))$ and the path via $v$, i.e. $d(p_h(v),p_\ell(v)) = d(p_h(v),v) + d(p_\ell(v),v) -\delta = 2\eps-\delta$. 
    The set of sites $S$ has size $n := 2|V|$.
    
    Suppose that we want to decide whether $G$ has a vertex cover of size $k$.
    Then the constructed polygonal domain and sites form an instance of \steinerspanner in which we allow $k$ Steiner points, and require a spanning ratio of $t = 3$ and a maximum complexity of $C = M \cdot (3n - 2k + 1/2)$.
    We show that this is a Yes-instance if and only if $G$ has a vertex cover of size at most $k$. As all coordinates of $P$ and the value of $C$ can be chosen such that they are polynomial in $n$, this implies that the \steinerspanner problem is NP-hard.
    
    \begin{figure}
        \centering
        \includegraphics{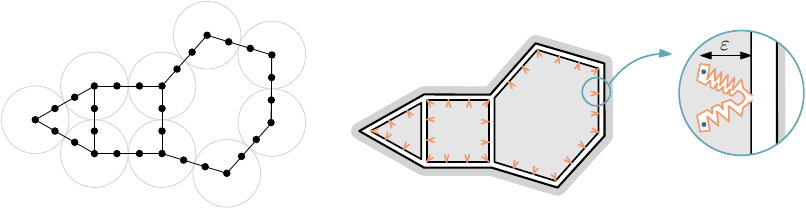}
        \caption{Constructing a \steinerspanner instance from a penny graph with subdivided edges.}
        \label{fig:np_hard_construction}
    \end{figure}

\subparagraph{Vertex cover implies spanner.} 
    Suppose there is a vertex cover $\VC$ of size at most $k$ for $G$. 
    For each vertex $v\in\VC$, we place a Steiner point $s(v)$ at the entrance of the corridors of $v$.
    Sites that belong to a vertex $v\in\VC$ lie at distance $\leq\eps$ from the Steiner point $s(v)$.
    The remaining sites in $S$ belong to a vertex adjacent to a vertex in $\VC$, and hence lie at distance $\leq 1+\eps$ from a Steiner point.
    
    To obtain a spanner $\G$ on the set $S$, we create, for each vertex $v\in V$, a link in $\G$ from its site $p_\ell(v)$ in its low complexity corridor, to its closest Steiner point, breaking ties arbitrarily.
    For the site $p_h(v)$ in the high complexity corridor of $v$, we add a link to $s(v)$ if $v \in \VC$, and to $p_\ell(v)$ otherwise. 
    We also add a link between every pair of Steiner points if that link does not visit another Steiner point.
    Let $p,q \in S$ be two sites. 
    To prove that~$\G$ has a spanning ratio of $3$, we consider three cases: (i) $p$ and $q$ belong to the same vertex, (ii) $p$ and $q$ belong to two adjacent vertices, or (iii) $p$ and $q$ belong to two non-adjacent vertices. 
    
    \begin{description}
        \item[(i) $p$ and $q$ belong to the same vertex $v$.] 
        We have $d(p,q) = 2\eps-\delta > \eps$.
        If there is no Steiner point at $v$, i.e. $v \notin \VC$, then the link $(p,q)$ is in $\G$. Otherwise, both links $(p,s(v))$ and $(q,s(v))$ of length $\eps$ are in $\G$. 
        Hence, $d_\G(p,q) \leq 2\eps \leq 3\eps < 3d(p,q)$.
        
        \item[(ii) $p$ and $q$ belong to adjacent vertices.]
        Then $d(p,q) = 1 + 2\eps$, and there must be a Steiner point at the entrance of the corridor of $p$ or $q$ or both. 
        Without loss of generality, assume there is a Steiner point $s(v)$ at the entrance of $p$. 
        From $q$ there is a path of length~$\leq 1 + 3\eps- \delta$ to its closest Steiner point.
        This Steiner point must have a path of length~$\leq 2$ to $s(v)$.
        Thus, the length of the path in the spanner connecting $p$ and $q$ is at most $d_\G(p,q) \leq 1 + 3\eps -\delta + 2 + \eps \leq 3 + 4\eps < 3 d(p,q)$.
    
        \item[(iii) $p$ and $q$ belong to non-adjacent vertices.] The distance between $p$ and $q$ is $\lambda + 2 \eps$ for some integer $\lambda \geq 2$. There is a path of length at most $\lambda + 2$ between the Steiner points $p$ and $q$ are connected to. Thus, $d_\G(p,q) \leq 1 + 3\eps -\delta + \lambda + 2 + 1 + 3\eps - \delta \leq 4 + \lambda + 6 \eps$. As $\lambda \geq 2$, $d_\G(p,q) \leq 3d(p,q)$.
    \end{description}
    
    Thus, $\G$ is a 3-spanner. 
    Finally, we need to bound the complexity of the spanner. 
    We first consider the links from a site $p_h(v)$ for all $v \in V$. 
    Such a link either has complexity $3M$ or $5M$, depending on whether there is a Steiner point $s(v)$. 
    The total complexity for all $v\in V$ is thus $5M(n/2 - k) + 3Mk = M(\frac{5}{2}n - 2k)$. 
    A link from a site $p_\ell(v)$ to a Steiner point has complexity as most $M+1$. 
    These links thus have complexity at most $(M+1) \cdot n/2$ in total. 
    As the maximum degree of a penny graph is six, and any vertex of degree greater than two is adjacent to only degree two vertices, the number of links from a Steiner point to other Steiner points can be at most six. 
    As each such link has complexity at most two, the total complexity of these links is at most $12k$. 
    The total complexity is thus at most $M(\frac{5}{2}n - 2k) + Mn/2 + n/2 + 12k \leq M(3n-2k) + 13 n \leq M(3n-2k+1/2)$, as $M > 26n$. 
    Thus, our spanner adheres to all of the posed conditions.

\subparagraph{Spanner implies vertex cover.}
    Suppose that $\G$ is a $3$-spanner on $S$ in $P$ that uses at most $k$ Steiner points and has complexity at most $C$.
    We show that $G$ then has a vertex cover of size at most $k$.
    
    \begin{observation}\label{obs:close_to_v}
        At least $k$ vertices are within distance $2\eps$ of a Steiner point.
    \end{observation}
    \begin{proof}
        For any vertex $v$, we have $d(p_h(v),p_\ell(v))=2\eps-\delta$.
        Let $\pi$ be a minimum length path in $\G$ that connects $p_h(v)$ to $p_\ell(v)$.
        Because $\G$ is a $3$-spanner, we have $\|\pi\|=d_\G(p_h(v),p_\ell(v))\leq 3d(p_h(v),p_\ell(v))\leq 6\eps-3\delta<6\eps$.
        Hence, the first half of $\pi$ lies within distance $3\eps$ of $p_h(v)$, and the second half lies within distance $3\eps$ of $p_\ell(v)$.
        Therefore, the entirety of $\pi$ lies within distance $2\eps$ of $v$, and if $p_h(v)$ and $p_\ell(v)$ are not connected by a direct link, there must exist a Steiner point within distance $2\eps$ of $v$.
        Suppose for a contradiction that there are at most $k-1$ vertices $v$ that are within distance $2\eps$ of a Steiner point.
        For those vertices $v$, the path in $\G$ connecting $p_h(v)$ to $p_\ell(v)$ has complexity at least $4M$, and these paths are disjoint for different vertices.
        For the (at least $n/2-k+1$) vertices $v$ that are not near a Steiner point, $\G$ contains a direct link from~$p_h(v)$ to~$p_\ell(v)$ of complexity~$5M$. Additionally, there exists a link that connects either~$p_h(v)$ or~$p_\ell(v)$ to another site to which a corridor of $v$ adds at least $M$ complexity.
        The total complexity of the spanner is then at least $(n/2-k + 1)(5M+M) + (k-1)4M = M(3n -2k + 2) > C$.
    \end{proof}
    
    \begin{observation}\label{obs:no_direct_edge}
        There is no link between any pair of sites $p,q$ that belong to different vertices and do not have a Steiner point within distance $2\eps$.
    \end{observation}
    \begin{proof}
        Let $v$ and $w$ be the vertices that $p$ and $q$ belong to, respectively. Suppose that there would be such a link. This link has complexity $\geq 2M$. Additionally, there must be a link from both $p$ and $q$ to its counterpart in the other corridor of $v$ or $w$ of complexity $5M$. As the spanner is connected, there must also be a path in the spanner from $p$ and $q$ to every other site. There must thus be another link of complexity $\geq M$ leaving a corridor of $v$ or $w$. To connect all other sites we still require at least $(n/2-2-k)5M + k3M + (n/2-2)M$ complexity. The total complexity is thus $\geq M(1 + 2(5 + 1) + 3n - 2k -10-2) = M(3n-2k+1)> C$, which is a contradiction.
    \end{proof}

        \begin{figure}
        \centering
        \includegraphics[page=4]{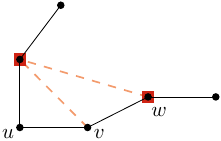}
        \caption{There are no 3- or 4-cycles in $G$.}
        \label{fig:proof_vc}
    \end{figure}
    To obtain a vertex cover $\VC$, we assign a vertex to $\VC$ if it has a Steiner point within distance $2\eps$. Observation~\ref{obs:close_to_v} tells us that $|\VC| \leq k$. To prove that $\VC$ is a vertex cover, assume for contradiction that there is an uncovered edge $(u,v) \in E$. Consider the two sites $p_\ell(u)$ and $p_\ell(v)$. First, note that there are no 3- and 4-cycles in $G$ by construction, thus the orange edges in Figure~\ref{fig:proof_vc} do not exist. Also, Observation~\ref{obs:no_direct_edge} implies that there is no direct edge between sites of $u$ and $v$, and thus  $p_\ell(u)$ and $p_\ell(v)$ both have a link to a Steiner point. Suppose the link between $p_\ell(u)$ and the Steiner point it is connected to does not visit $v$, and the link between $p_\ell(v)$ and the Steiner point it is connected to does not visit $u$. 
    Then $d_\G(p_\ell(u),p_\ell(v))\geq \eps + 2(1-2\eps) + 1 + 2(1-2\eps) + \eps = 5 - 6\eps$. This contradicts the spanning ratio, as $d(p_\ell(u),p_\ell(v)) = 1+2\eps$. Next, suppose without loss of generality that the link of $p_\ell(v)$ does visit $u$ (and the link of $p_\ell(u)$ does not visit $v$). In that case, consider another neighbor $w \neq u$ of $v$. Because the distance from $v$ to the Steiner point is at least $2-2\eps$, and there are no short cycles, we have that $d_\G(p_\ell(v),p_\ell(w)) \geq \eps + 2(2 -2\eps) + \eps = 4 -2\eps> 3d(p_\ell(v),p_\ell(w))$. This again contradicts the spanning ratio. We conclude that $\VC$ is a vertex cover for $G$.

    \begin{theorem}
      \steinerspanner is NP-hard. More precisely, it is NP-hard to decide whether a polygonal domain with $n$ sites admits a $3$-spanner using at most $k$ Steiner points and complexity at most $C = M \cdot (3n - 2k + 1/2)$ with $M> 26 n$.
    \end{theorem}

\section{Future work}\label{sec:future_work}
On the side of constructing low-complexity spanners, an interesting direction for future work would be to close the gap between the upper and lower bounds, both with and without using Steiner points. We believe it might be possible to increase the $n^{1/(t+1)}$ term to $n^{1/t}$ (or even $n^{1/(t-1)}$) in 
Lemma~\ref{lem:lowerbound_tspanner}. On the side of the hardness, many interesting open questions remain, such as: Is the problem still hard in a simple polygon? Can we show hardness for other spanning ratios and/or a less restricted complexity requirement? Is the problem even in NP?

\bibliography{./bibliography.bib}

\begin{thebibliography}{10}

\bibitem{SpannerPolygonalDomain}
Mohammad~Ali Abam, Marjan Adeli, Hamid Homapour, and Pooya~Zafar Asadollahpoor.
\newblock Geometric spanners for points inside a polygonal domain.
\newblock In {\em Proc. 31st International Symposium on Computational Geometry,
  SoCG}, volume~34 of {\em LIPIcs}, pages 186--197. Schloss Dagstuhl -
  Leibniz-Zentrum f{\"{u}}r Informatik, 2015.

\bibitem{SpannerPolyhedralTerrain}
Mohammad~Ali Abam, Mark de~Berg, and Mohammad Javad~Rezaei Seraji.
\newblock Geodesic spanners for points on a polyhedral terrain.
\newblock {\em {SIAM} J. Comput.}, 48(6):1796--1810, 2019.

\bibitem{alzoubi03geomet_spann_wirel_ad_hoc_networ}
Khaled~M. Alzoubi, Xiang{-}Yang Li, Yu~Wang, Peng{-}Jun Wan, and Ophir Frieder.
\newblock Geometric spanners for wireless ad hoc networks.
\newblock {\em {IEEE} Trans. Parallel Distributed Syst.}, 14(4):408--421, 2003.

\bibitem{Arya95short_thin_lanky}
Sunil Arya, Gautam Das, David~M. Mount, Jeffrey~S. Salowe, and Michiel H.~M.
  Smid.
\newblock Euclidean spanners: short, thin, and lanky.
\newblock In {\em Proc. 27th Annual {ACM} Symposium on Theory of Computing,
  {STOC}}, pages 489--498. {ACM}, 1995.

\bibitem{bhore21light_euclid_stein_spann_plane}
Sujoy Bhore and Csaba~D. T{\'{o}}th.
\newblock Light {Euclidean Steiner} spanners in the plane.
\newblock In {\em Proc. 37th International Symposium on Computational Geometry,
  {SoCG}}, volume 189 of {\em LIPIcs}, pages 15:1--15:17. Schloss Dagstuhl -
  Leibniz-Zentrum f{\"{u}}r Informatik, 2021.

\bibitem{borradaile15near_stein}
Glencora Borradaile and David Eppstein.
\newblock Near-linear-time deterministic plane {Steiner} spanners for
  well-spaced point sets.
\newblock {\em Comput. Geom.}, 49:8--16, 2015.

\bibitem{Bose05Bounded_degree_low_weight}
Prosenjit Bose, Joachim Gudmundsson, and Michiel H.~M. Smid.
\newblock Constructing plane spanners of bounded degree and low weight.
\newblock {\em Algorithmica}, 42(3-4):249--264, 2005.

\bibitem{survey_geometric_spanners}
Prosenjit Bose and Michiel H.~M. Smid.
\newblock On plane geometric spanners: {A} survey and open problems.
\newblock {\em Comput. Geom.}, 46(7):818--830, 2013.

\bibitem{vertexcover_penny}
M{\'{a}}rcia~R. Cerioli, Lu{\'{e}}rbio Faria, Talita~O. Ferreira, and
  F{\'{a}}bio Protti.
\newblock A note on maximum independent sets and minimum clique partitions in
  unit disk graphs and penny graphs: complexity and approximation.
\newblock {\em {RAIRO} Theor. Informatics Appl.}, 45(3):331--346, 2011.

\bibitem{RoutingInDoublingMetrics}
T.{-}H.~Hubert Chan, Anupam Gupta, Bruce~M. Maggs, and Shuheng Zhou.
\newblock On hierarchical routing in doubling metrics.
\newblock {\em {ACM} Trans. Algorithms}, 12(4):55:1--55:22, 2016.

\bibitem{Chan15Doubling_spanners}
T.{-}H.~Hubert Chan, Mingfei Li, Li~Ning, and Shay Solomon.
\newblock New doubling spanners: Better and simpler.
\newblock {\em {SIAM} J. Comput.}, 44(1):37--53, 2015.

\bibitem{Chazelle_triangulate}
Bernard Chazelle.
\newblock Triangulating a simple polygon in linear time.
\newblock {\em Discret. Comput. Geom.}, 6:485--524, 1991.

\bibitem{Clarkson87_approx_sp}
Kenneth~L. Clarkson.
\newblock Approximation algorithms for shortest path motion planning.
\newblock In {\em Proc. 19th Annual {ACM} Symposium on Theory of Computing,
  {STOC}}, pages 56--65. {ACM}, 1987.

\bibitem{cohen-addad20light_spann_low_embed_effic}
Vincent Cohen{-}Addad, Arnold Filtser, Philip~N. Klein, and Hung Le.
\newblock On light spanners, low-treewidth embeddings and efficient traversing
  in minor-free graphs.
\newblock In {\em Proc. 61st {IEEE} Annual Symposium on Foundations of Computer
  Science, {FOCS}}, pages 589--600. {IEEE}, 2020.

\bibitem{CyganFKLMPPS15parameterized}
Marek Cygan, Fedor~V. Fomin, Lukasz Kowalik, Daniel Lokshtanov, D{\'{a}}niel
  Marx, Marcin Pilipczuk, Michal Pilipczuk, and Saket Saurabh.
\newblock {\em Parameterized Algorithms}.
\newblock Springer, 2015.

\bibitem{full_version}
Sarita de~Berg, Marc van Kreveld, and Frank Staals.
\newblock The complexity of geodesic spanners.
\newblock {\em CoRR}, abs/2303.02997, 2023.

\bibitem{complexity_spanners}
Sarita de~Berg, Marc~J. van Kreveld, and Frank Staals.
\newblock The complexity of geodesic spanners.
\newblock In {\em Proc. 39th International Symposium on Computational Geometry,
  {SoCG}}, volume 258 of {\em LIPIcs}, pages 16:1--16:16. Schloss Dagstuhl -
  Leibniz-Zentrum f{\"{u}}r Informatik, 2023.

\bibitem{shallow_low_light_trees}
Yefim Dinitz, Michael Elkin, and Shay Solomon.
\newblock Shallow-low-light trees, and tight lower bounds for {Euclidean}
  spanners.
\newblock In {\em Proc. 49th Annual {IEEE} Symposium on Foundations of Computer
  Science, {FOCS}}, pages 519--528, 2008.

\bibitem{Elkin_low_light_Steiner}
Michael Elkin and Shay Solomon.
\newblock Narrow-shallow-low-light trees with and without {Steiner} points.
\newblock {\em {SIAM} J. Discret. Math.}, 25(1):181--210, 2011.

\bibitem{Elkin15optimal_short_thin_lanky}
Michael Elkin and Shay Solomon.
\newblock Optimal {Euclidean} spanners: Really short, thin, and lanky.
\newblock {\em J. {ACM}}, 62(5):35:1--35:45, 2015.

\bibitem{gottlieb17effic_regres_metric_spaces_approx_lipsc_exten}
Lee{-}Ad Gottlieb, Aryeh Kontorovich, and Robert Krauthgamer.
\newblock Efficient regression in metric spaces via approximate {Lipschitz}
  extension.
\newblock {\em {IEEE} Trans. Inf. Theory}, 63(8):4838--4849, 2017.

\bibitem{bounded_doubling_optimal_dynamic}
Lee{-}Ad Gottlieb and Liam Roditty.
\newblock An optimal dynamic spanner for doubling metric spaces.
\newblock In {\em Proc. 16th Annual European Symposium on Algorithms, {ESA}},
  volume 5193 of {\em LNCS}, pages 478--489, 2008.

\bibitem{2PSP_simple_polygon}
Leonidas~J. Guibas and John Hershberger.
\newblock Optimal shortest path queries in a simple polygon.
\newblock {\em J. Comput. Syst. Sci.}, 39(2):126--152, 1989.

\bibitem{FastConstructionDoublingMetrics}
Sariel Har{-}Peled and Manor Mendel.
\newblock Fast construction of nets in low-dimensional metrics and their
  applications.
\newblock {\em {SIAM} J. Comput.}, 35(5):1148--1184, 2006.

\bibitem{SPMHershbergerSuri}
John Hershberger and Subhash Suri.
\newblock An optimal algorithm for {Euclidean} shortest paths in the plane.
\newblock {\em {SIAM} J. Comput.}, 28(6):2215--2256, 1999.

\bibitem{Kirkpatrick_point_location}
David~G. Kirkpatrick.
\newblock Optimal search in planar subdivisions.
\newblock {\em {SIAM} J. Comput.}, 12(1):28--35, 1983.

\bibitem{Le19truly_opt_Euclidean}
Hung Le and Shay Solomon.
\newblock Truly optimal {Euclidean} spanners.
\newblock In {\em Proc. 60th {IEEE} Annual Symposium on Foundations of Computer
  Science, {FOCS}}, pages 1078--1100. {IEEE} Computer Society, 2019.

\bibitem{Levcopoulos98_fault_tolerant}
Christos Levcopoulos, Giri Narasimhan, and Michiel H.~M. Smid.
\newblock Efficient algorithms for constructing fault-tolerant geometric
  spanners.
\newblock In {\em Proc. 13th Annual {ACM} Symposium on the Theory of Computing,
  {STOC}}, pages 186--195. {ACM}, 1998.

\bibitem{proximity_algorithms_book}
Joseph S.~B. Mitchell and Wolfgang Mulzer.
\newblock Proximity algorithms.
\newblock In {\em Handbook of Discrete and Computational Geometry (3rd
  Edition)}, chapter~32, pages 849--874. {Chapman \& Hall/CRC}, 2017.

\bibitem{book_spanners}
Giri Narasimhan and Michiel H.~M. Smid.
\newblock {\em Geometric Spanner Networks}.
\newblock Cambridge University Press, 2007.

\bibitem{remy10euclid}
Jan Remy, Reto Sp{\"{o}}hel, and Andreas Wei{\ss}l.
\newblock On {Euclidean} vehicle routing with allocation.
\newblock {\em Comput. Geom.}, 43(4):357--376, 2010.

\end{thebibliography}

\end{document}